\documentclass{llncs}

\usepackage{enumerate}
\usepackage[all]{xy}
\usepackage{graphicx}

\usepackage{amssymb, latexsym, amscd,amsmath,stmaryrd,wrapfig}
\usepackage{hyperref}
\usepackage{wasysym}
\begin{document}

\author{Tomasz Brengos \thanks{This work has been supported by the grant of Warsaw University of Technology no. 504M for young researchers.
}}
\title{On coalgebras with internal moves}

\authorrunning{T. Brengos}

\institute{Faculty of Mathematics and Information Science\\
         Warsaw University of Technology\\ 
         Koszykowa 75 \\       
         00-662 Warszawa, Poland \\ \email{t.brengos@mini.pw.edu.pl}}

\maketitle

\begin{abstract}
 In the first part of the paper we recall the coalgebraic approach to handling the so-called invisible transitions that appear in different state-based systems semantics. We claim that these transitions are always part of the unit of a certain monad. Hence, coalgebras with internal moves are exactly coalgebras over a monadic type. The rest of the paper is devoted to supporting our claim by studying two important behavioural equivalences for state-based systems with internal moves, namely: weak bisimulation and trace semantics.
 
We continue our research on weak bisimulations for coalgebras over order enriched monads. The key notions used in this paper and proposed by us in our previous work are the notions of an order saturation monad and a saturator.  A saturator operator can be intuitively understood as a reflexive, transitive closure operator.  There are two approaches towards defining saturators for coalgebras with internal moves. Here, we give necessary conditions for them to yield the same notion of weak bisimulation.

Finally, we propose a definition of trace semantics for coalgebras with silent moves via a  uniform fixed point operator. We compare strong and weak bisimilation together with trace semantics for coalgebras with internal steps.
\end{abstract}

\section{Introduction}
In recent years we have witnessed a rapid development of the theory of coalgebras as a unifying theory for state-based systems \cite{GummEl,HasJacSok,JacSilSok,Rutten}. Coalgebras to some extent are one-step entities in their nature. They can be thought of and understood as a representation of a single step of visible computation of a given process.  Yet, for many state-based systems it is useful to consider a part of computation branch that is allowed to take several steps and in some sense remains neutral (invisible) to the structure of the process. For instance, the so-called \emph{$\tau$-transitions} also called \emph{invisible transitions} for labelled transition systems \cite{Milner,Milner3} or \emph{$\varepsilon$-transitions} for non-deterministic automata \cite{HopUll}. As will be witnessed here, these special branches of computation are the same in their nature, yet they are used in order to develop different notions of equivalence of processes, e.g. weak bisimulation for LTS \cite{Milner} or trace semantics for non-deterministic automata with $\varepsilon$-moves, we call $\varepsilon$-NA in short \cite{HopUll}. These are not the only state-based systems  considered in the literature with a special invisible computational branch. Fully probabilistic systems \cite{BaierHermanns} or Segala systems \cite{Segala,SegalaThesis} are among those, to name a few. All these systems are instances of a general notion of a coalgebra. If so, then how should we consider the invisible part of computation coalgebraically? As we will see further on, the invisible part of the computation can be and should be, in our opinion, considered as part of the unit of a monad. Before we state basic results let us summarize known literature 
on the topic of invisible transitions from perspective of weak bisimulation, trace semantics and coalgebra.
\subsubsection{Weak bisimulation}
The notion of a strong bisimulation for different transition systems plays an important role in theoretical computer science. A weak bisimulation is a relaxation of this notion by allowing silent, unobservable transitions. Here, we focus on the weak bisimulation and weak bisimilarity proposed by R.~Milner \cite{Milner,Milner3} (see also \cite{Sangiorgi11}). Analogues of Milner's weak bisimulation are established for different deterministic and probabilistic transition systems (e.g. \cite{BaierHermanns,Sangiorgi11,Segala,SegalaThesis}). It is well known that one can introduce Milner's weak bisimulation for LTS in several different but equivalent ways. 

The notion of a strong bisimulation, unlike the weak bisimulation, has been well captured coalgebraically (see e.g. \cite{GummEl,Rutt2000,Staton}).  Different approaches to defining weak bisimulations for coalgebras have been presented in the literature. The earliest paper is \cite{Rutten}, where the author studies weak bisimulations for while programs. In \cite{Rothe} the author introduces a definition of weak bisimulation for coalgebras by translating a coalgebraic structure into an LTS. This construction works for coalgebras over a large class of functors but does not cover the distribution functor, hence it is not applicable to different types of probabilistic systems. In \cite{RotheMas} weak bisimulations are introduced via weak homomorphisms. As noted in \cite{SokViWor} this construction does not lead to intuitive results for probabilistic systems. In \cite{SokViWor} the authors present a definition of weak bisimulation for classes of coalgebras over functors obtained from bifunctors. Here, weak bisimulation of a system is defined as a strong bisimulation of a transformed system. In \cite{Brengos12} we proposed a new approach to defining weak bisimulation in two different ways. Two definitions of weak bisimulation described by us in \cite{Brengos12} were proposed in the setting of coalgebras over ordered functors.  The key ingredient of the definitions is the notion of a saturator. As noted in \cite{Brengos12} the saturator is sometimes too general to model only weak bisimulation and may be used to define other known equivalences, e.g. delay bisimulation \cite{Sangiorgi11}. Moreover, the saturators from \cite{Brengos12} do not arise in any natural way. To deal with this problem we have presented a canonical way to consider weak bisimulation saturation in our previous paper \cite{Brengos13}. Part of the results from \cite{Brengos13} are recalled in this paper. We recall the two procedures for handling the invisible part of computation via monadic structure, the definition and properties of order saturation monad \cite{Brengos13}. What is new here is the definition of weak bisimulation in terms of a kernel bisimulation on a saturated model and the comparison of the two strategies towards saturation from the point of view of weak bisimulation.  

Here, we should also mention \cite{GonPat,MicPer} which appeared almost at the same time as our previous paper \cite{Brengos13}. The former is a talk on the on-going research by S.~Goncharov and D.~Pattinson related to weak bisimulation for coalgebras (at the time of writing this paper we were unable to find any other reference to their work). Their approach is very similar to ours as it uses ordered monads and fixed points for saturations. However, the authors do not hide the invisible steps inside a monadic structure. The latter is a paper in which the authors study weak bisimulation for labelled transition systems weighted over semirings. They propose a coalgebraic approach towards defining  weak bisimulation which relies on $\varepsilon$-elimination procedure presented in \cite{SilWester}.

\subsubsection{Weak trace semantics}
Trace semantics is a standard behavioural equivalence for many state-based systems. 
Generic trace semantics for coalgebras has been proposed in \cite{HasJacSok,JacSilSok}. If $T$ is a monad on a category $\mathsf{C}$ and $F:\mathsf{C}\to\mathsf{C}$ is an endofunctor then the trace semantics of $TF$-coalgebras is final semantics for coalgebras  considered in a different category, namely the Kleisli category for the monad $T$ \cite{HasJacSok,Jacobs08}. It is worth noting that trace semantics can also be defined for $GT$-coalgebras for an endofunctor $G:\mathsf{C}\to\mathsf{C}$ \cite{JacSilSok,SilBonBonRut} via the so-called $\mathcal{EM}$-extension semantics. In our paper however, we focus only on $TF$-coalgebras and do not consider $GT$-coalgebras. Trace semantics can also be defined for different state-based systems with internal, invisible moves. In order to distinguish trace semantics for systems with and  without silent steps  we will sometimes call the former ``weak trace semantics".
One coalgebraic approach towards defining trace semantics for systems with $\varepsilon$-moves (invisible moves) is based on a very simple idea, has been presented in \cite{HasJacSok_jap,SilWester} and can be summarized as follows. In the first step we consider invisible moves as visible. Then we find the trace semantics for an "all-visible-steps" coalgebra and finally, we remove all occurrences of the invisible label and get the desired weak trace semantics.  We discuss this approach in our paper and call it the ``top-down" approach. The term ``top-down" refers to the fact that we somewhat artificially treat the invisible moves as if they were visible  and then we remove their occurrences from the trace. Such an approach does not use any structural properties of silent moves. A dual approach, a ``bottom-up" method, should make use of their structural properties.  Here, we present a ``bottom-up" method for coalgebras with internal steps that treats silent moves as part of the unit of a certain monad. 

\subsubsection{Content and organization of the paper}
The paper is organized as follows. Section \ref{section:introduction} recalls basic notions in category theory, algebra and coalgebra. Section~\ref{subsection:monads} describes two very general methods for dealing with silent steps via a monadic structure that have been proposed in our previous work \cite{Brengos13}. We will see that these two methods appear in classical definitions of a weak bisimulation for LTS's.  In Section~\ref{section:LTS_coalgebraically} we recall the definition of an order saturation monad that comes from \cite{Brengos13} and claim that this object is suitable for defining weak bisimulations for coalgebras.  An order saturation monad is an order enriched monad equipped with an extra operator, a saturator $(-)^{*}$, that assigns to any coalgebra $\alpha:X\to TX$ a coalgebra $\alpha^{*}~:~X\to TX$ and can be thought of as a reflexive, transitive closure operator.  It turns out that in the classical literature on labelled transition systems and weak bisimulation one can find two different saturators yelding the same notion of equivalence. These two saturators are natural consequences of the two stategies towards handling invisible steps via monadic structure.  What is new in this section is the following: 
\begin{itemize}
\item Weak bisimulation is defined as a kernel bisimulation \cite{Staton} on a saturated structure and not via lax- and oplax-homomorphisms in Aczel-Mendler style as it was done in \cite{Brengos13}.  
\item \label{pt2} We present both saturators in a general setting and ask when they yield the same notion of weak bisimulation. We give sufficient conditions functors should satisfy so that weak bisimulation coincides for both approaches. 
\end{itemize}
 In Section~\ref{section:trace_semantics} we discuss a novel approach towards defining trace semantics for coalgebras with internal moves. Here, weak trace semantics morphism is obtained axiomatically by the so-called coalgebraic trace operator, i.e. a uniform fixed point operator. For $\textbf{Cppo}$-enriched monads, a coalgebraic trace operator is given by the least fixed point operator $\mu x.(x\cdot \alpha)$.  Moreover, we show that the coalgebraic trace operator for $\varepsilon$-NA's arises from properties of the so-called free LTS monad. To be more precise,  Kleisli category for the free LTS monad is traced monoidal category in the sense of Joyal et al. \cite{JSV}. In Section \ref{section:weak_and_trace}, in a fairly general setting, we formulate how strong bisimulation, weak bisimulation and weak trace semantics are related. Hence, according to our knowledge we present the first paper that considers a comparison of three different behaviour equivalences in van Glabbeek's spectrum for systems with internal moves \cite{Glabbeek} from coalgebraic perspective.

\section{Basic notions and properties} 
\label{section:introduction}

\subsubsection{Algebras and coalgebras} 
Let $\mathsf{C}$ be a category and let $F \colon \mathsf{C}\rightarrow \mathsf{C}$ be a functor.  An $F$-algebra is a morphism $a:FA\to A$ in $\mathsf{C}$. A \emph{homomorphism} between algebras $a:FA\to A$ and $b:FB\to B$ is a morphism $f:A\to B$ in $\mathsf{C}$ such that $b\circ F(f) = f\circ a$. Dually, an \emph{$F$-coalgebra} is a morphism $\alpha:X\to FX$ in $\mathsf{C}$.  The domain $X$ of $\alpha$ is called \emph{carrier} and the morphism  $\alpha$ is sometimes also called \emph{structure}. A \emph{homomorphism} from an $F$-coalgebra $\alpha:X\to FX$ to an $F$-coalgebra $\beta:Y\to FY$  is a morphism $f \colon X\rightarrow Y$ in $\mathsf{C}$ such that $ F(f)\circ \alpha =  \beta \circ f$. The category of all $F$-coalgebras ($F$-algebras) and homomorphisms between them is denoted by $\mathsf{C}_F$ (resp. $\mathsf{C}^F$). Many transition systems can be captured by the notion of coalgebra. In this paper we  mainly focus on labelled transition systems with a silent label and non-deterministic automata with $\varepsilon$-moves. These  two structures have been defined and thoroughly studied in the computer science literature (see e.g. \cite{HopUll,Milner,Milner3,Sangiorgi11}). Let $\Sigma$ be a fixed set of alphabet letters. A \emph{labelled transition system} over the alphabet $\Sigma_\tau = \Sigma +  \{\tau\}$ (or an \emph{LTS} in short) is a triple $\left<X,\Sigma_\tau,\to\right>$, where $X$ is called a \emph{set of states} and $\to\subseteq X\times \Sigma_\tau\times X$ is a \emph{transition}. The label $\tau$ is considered a special label sometimes called silent or invisible label. For an LTS $\left<X,\Sigma_\tau,\to\right>$ instead of writing $(x,\sigma,x')\in \to$ we write $x\stackrel{\sigma}{\to}x'$. Labelled transition systems can be viewed as coalgebras over the type $\mathcal{P}(\Sigma_\tau\times \mathcal{I}d)$ \cite{Rutten}. 
From coalgebraic perspective, a \emph{non-deterministic automaton with $\varepsilon$- transitions}, or \emph{$\varepsilon$-NA} in short, over alphabet $\Sigma$ is a coalgebra of the type $\mathcal{P}(\Sigma_\varepsilon\times \mathcal{I}d +  1)$, where $1=\{\checked\}$ is  fixed one element set and $\Sigma_\varepsilon = \Sigma +  \{\varepsilon\}$. Note that LTS's differ from $\varepsilon$-NA's in the presence of $1$ in the type.  It is responsible for specifying which states are \emph{final} and which are not. To be more precise for $\varepsilon$-NA $\alpha:X\to \mathcal{P}(\Sigma_\varepsilon \times X +  1)$ we call a state $x\in X$ \emph{final} if $\checked \in \alpha(x)$. For more information on automata the reader is referred to e.g. \cite{HopUll}. 
\subsubsection{Strong bisimulation for coalgebras}
Notions of strong bisimulation have been well captured coalgebraically \cite{AczMen,GummEl,Rutten,Staton}.  Let $F$ be a $\mathsf{Set}$-endofunctor and consider an $F$-coalgebra $\alpha:X\to FX$. In Aczel-Mendler style \cite{AczMen,Staton}, \emph{a (strong) bisimulation} is a relation $R\subseteq X\times X$ for which there is a structure $\gamma:R\to TR$ making $\pi_1:R\to X$ and $\pi_2:R\to X$ homomorphisms between $\gamma$ and $\alpha$. 
In this paper however we consider defining bisimulation as the so-called kernel bisimulation \cite{Staton}. Let $F:\mathsf{C}\to\mathsf{C}$ be an endofunctor on an arbitrary category. Let $\alpha:X\to FX$ and $\beta:Y\to FY$ be $F$-coalgebras. A relation $R$ on $X$ and $Y$ (i.e. a jointly-monic span $X\stackrel{\pi_1}{\leftarrow} R \stackrel{\pi_2}{\rightarrow} Y$ in $\mathsf{C}$) is  \emph{kernel bisimulation} or \emph{bisimulation} in short if there is a coalgebra $\gamma:Z\to FZ$ and homomorphisms $f$ from $\alpha$ to $\gamma$ and $g$ from $\beta$ to $\gamma$ such that   $R$ with $\pi_1$, $\pi_2$ is the pullback of $X\stackrel{f}{\to} Z \stackrel{g}{\leftarrow} Y$. For a thorough study of the relation between Aczel-Mendler style of defining bisimulation and kernel bisimulation the reader is referred to \cite{Staton} for details. 

\subsubsection{Monads} \begin{wrapfigure}[9]{r}{0.15\textwidth}
\vspace{-40pt}
$$
\xymatrix@-1pc{
T^3\ar[d]_{T\mu} \ar[r]^{\mu_T} & T^2\ar[d]^{\mu}  \\
T^2 \ar[r]_{\mu} & T 
}
$$
$$
\xymatrix@-1pc{
T\ar[d]_{\eta_T}\ar[r]^{T\eta}\ar@{=}[dr] & T^2\ar[d]^{\mu} \\
T^2 \ar[r]_\mu & T
}
$$
\end{wrapfigure}
A \emph{monad} on $\mathsf{C}$ is a triple $(T,\mu,\eta)$, where $T:\mathsf{C}\to \mathsf{C}$ is an endofunctor and $\mu:T^2\implies T$, $\eta:\mathcal{I}d\implies T$ are two natural transformations for which the following two diagrams commute:
The transformation $\mu$ is called  \emph{multiplication} and $\eta$ \emph{unit}.
Each monad gives rise to a canonical category - Kleisli category for $T$. If $(T,\mu,\eta)$ is a monad on category $\mathsf{C}$ then \emph{Klesli category} $\mathcal{K}l(T)$ for $T$ has the class of objects equal to the class of objects of $\mathsf{C}$ and for two objects $X,Y$ in $\mathcal{K}l(T)$ we have $Hom_{\mathcal{K}l(T)}(X,Y) = Hom_{\mathsf{C}}(X,TY)$
with the composition $\cdot$ in $\mathcal{K}l(T)$ defined between two morphisms $f:X\to TY$ and $g:Y\to TZ$ by
$g\cdot f := \mu_Z \circ T(g) \circ f$ (here, $\circ$ denotes the composition in $\mathsf{C}$).

\begin{example}
The powerset endofunctor $\mathcal{P}:\mathsf{Set}\to \mathsf{Set}$ is a monad with the multiplication $\mu:\mathcal{P}^2\implies\mathcal{P}$ and the unit $\eta:\mathcal{I}d\implies \mathcal{P}$  given on their $X$-components by 
$\mu_X:\mathcal{P}\mathcal{P}X\to \mathcal{P}X; S \mapsto \bigcup S$  and $\eta_X:X\to \mathcal{P}X; x\mapsto \{x\}$.
For any category $\mathsf{C}$ with binary coproducts and an object $A\in \mathsf{C}$ define $\mathcal{M}_A:\mathsf{C}\to\mathsf{C}$ as $\mathcal{M}_A =\mathcal{I}d +  A$. The functor carries a monadic structure $(\mathcal{M}_A,\mu,\eta)$, where the $X$-components of the multiplication and the unit are the following: $\mu_X:(X+A)+A\to X+A; \mu_X = [id_{X+A},\iota^2]$ and $\eta_X:X\to X+A; \eta_X=\iota^1$. Here, $\iota^1$ and $\iota^2$ denote the coprojections into the first and the second component of $X+A$ respectively.
The monad $\mathcal{M}_A$ is sometimes called \emph{exception monad}.
\end{example}

Since in many cases we will work with two categories at once: $\mathsf{C}$ and $\mathcal{K}l(T)$, morphisms in $\mathsf{C}$ will be denoted using standard arrow $\to$, whereas for morphisms in $\mathcal{K}l(T)$ we will use the symbol $\multimap$. For any object $X$ in $\mathsf{C}$ (or equivalently in $\mathcal{K}l(T)$) the identity map from $X$ to itself in $\mathsf{C}$  will be denoted by $id_X$ and in $\mathcal{K}l(T)$ by $1_X$ or simply $1$ if the domain can be deduced from the context. 

The category $\mathsf{C}$ is a subcategory of $\mathcal{K}l(T)$ where the inclusion functor $^{\sharp}$ sends each object $X\in \mathsf{C}$ to itself and each morphism $f:X\to Y$ in $\mathsf{C}$ to the morphism $f^{\sharp}:X\multimap Y$ given by $
f^{\sharp}:X\to TY; f^{\sharp} = \eta_Y\circ f$.
Each monad $(T,\mu,\eta)$ on a category $\mathsf{C}$ arises as the composition of left and right adjoint:
\begin{wrapfigure}[10]{r}{0.21\textwidth}
\vspace{-40pt}
$$
\xymatrix@-1pc{
\mathsf{C}\ar@/^1.5pc/[r]^{^{\sharp}}\ar@{}[r]|\perp & \mathcal{K}l(T)\ar@/^1.5pc/[l]^{U_T}
\\
\\
\mathcal{K}l(T) \ar[r]^{\overline{F}} & \mathcal{K}l(T)  \\
\mathsf{C}\ar[u]^{\sharp} \ar[r]_F & \mathsf{C}\ar[u]_{\sharp}
}
$$
\end{wrapfigure}
Here, $U_T:\mathcal{K}l(T)\to \mathsf{C}$ is a functor defined as follows. For any object $X\in \mathcal{K}l(T)$ (i.e. $X\in \mathsf{C}$) the object $U_T X$ is given by $U_T X := TX$ and for any morphism $f:X\multimap Y$ in $\mathcal{K}l(T)$ (i.e. $f:X\to TY$ in $\mathsf{C}$) the morphism $U_T f:TX\to TY$ is given by $U_T f = \mu_Y\circ Tf$.

We say that a functor $F:\mathsf{C}\to\mathsf{C}$ \emph{lifts to} an endofunctor $\overline{F}:\mathcal{K}l(T)\to\mathcal{K}l(T)$ provided that the following diagram commutes \cite{HasJacSok,JacSilSok}:

\noindent There is a one-to-one correspondence between liftings $\overline{F}$ and \emph{distributive laws} $\lambda: FT\implies TF$  \cite{JacSilSok,ManMul}.
Given a distributive law $\lambda:FT\implies TF$ a lifting $\overline{F}:\mathcal{K}l(T)\to\mathcal{K}l(T)$ is defined by:
\begin{align*}
& \overline{F}X := FX \text{ for any object } X\in \mathcal{K}l(T),\\
& \overline{F}f:FX\to TFY; \overline{F}f = \lambda_Y \circ Ff \text{ for any morphism } f:X\to TY.
\end{align*}
Conversely, a lifting $\overline{F}:\mathcal{K}l(T)\to \mathcal{K}l(T)$ of $F$ gives rise to a distributive law $\lambda:FT\implies TF$ defined by $\lambda_X:FTX\to TFX; \lambda_X = \overline{F}(id_{TX})$. 
A monad $T$ on a cartesian closed category $\mathsf{C}$ is called \emph{strong} if there is a transformation $\text{st}_{X,Y}:X\times TY\to T(X\times Y)$ called \emph{tensorial strength} satisfying the strength laws listed in e.g. \cite{Kock}. Existence of strength guarantees that for any object $\Sigma$ the functor $\Sigma\times \mathcal{I}d:\mathsf{C}\to \mathsf{C}$ admits a lifting  $\overline{\Sigma}:\mathcal{K}l(T)\to\mathcal{K}l(T)$.
To be more precise we define a functor $\overline{\Sigma}:\mathcal{K}l(T)\to\mathcal{K}l(T)$ as follows. For any object $X\in \mathcal{K}l(T)$ (i.e. $X\in \mathsf{C}$) we put
$
\overline{\Sigma}X := \Sigma\times X,
$
and for any morphism $f:X\multimap Y$ (i.e. $f:X\to TY$ in $\mathsf{C}$) we define $\overline{\Sigma}f:\Sigma\times X\to T(\Sigma\times Y)$ by $\overline{\Sigma}f:= \text{st}_{\Sigma,Y}\circ ( id_{\Sigma}\times f)$.
Existence of the transformation $\text{st}_{X,Y}$ is not a strong assumption. For instance all monads on $\mathsf{Set}$ are strong.

A category is \emph{order enriched} if each hom-set is a poset with order preserved by composition. An endofunctor  on an order enriched category is \emph{locally monotonic} if it preserves order. 
A category $\mathsf{C}$ is \emph{$\textbf{Cppo}$-enriched} if  for any objects $X,Y$:
\begin{itemize}
\item the hom-set $Hom_{\mathsf{C}}(X,Y)$ is a poset with a least element $\perp$, 
\item for any ascending $\omega$-chain $f_0\leqslant f_1\leqslant \ldots $ in  $Hom_{\mathsf{C}}(X,Y)$ the supremum $\bigvee_{i\in \mathbb{N}} f_i$ exists,
\item $g\circ \bigvee_{i\in \mathbb{N}} f_i = \bigvee_{i\in \mathbb{N}} g\circ f_i$ and $(\bigvee_{i\in \mathbb{N}} f_i)\circ h = \bigvee_{i\in \mathbb{N}} f_i\circ h$ for any ascending $\omega$-chain $f_0\leqslant f_1\leqslant \ldots$ and $g,h$ with suitable domain and codomain.
\end{itemize}
Note that it is \emph{not} necessarily the case that $f\circ \perp=\perp$ or $\perp\circ f=\perp$ for any morphism $f$.
 An endofunctor on a $\textbf{Cppo}$-enriched category is called \emph{locally continuous} if it preserves suprema of ascending $\omega$-chains. For more details on $\bf{Cppo}$-enriched categories the reader is referred to e.g. \cite{AbraJung,HasJacSok}.
\begin{example}
The Kleisli category for the powerset monad $\mathcal{P}$ is $\textbf{Cppo}$-enriched \cite{HasJacSok}. The order on the hom-sets is imposed by the natural point-wise order. The strength map for $\mathcal{P}$ is given by
$$\text{st}_{X,Y}:X\times \mathcal{P}Y\to \mathcal{P}(X\times Y); (x,S) \mapsto \{(x,y)\mid y\in S\}.$$ 
The lifting $\overline{\Sigma}:\mathcal{K}l(\mathcal{P})\to \mathcal{K}l(\mathcal{P})$ of $\Sigma\times \mathcal{I}d:\mathsf{Set}\to \mathsf{Set}$ is a locally continuous functor \cite{HasJacSok}.
The Kleisli category for the monad $\mathcal{M}_1$ on $\mathsf{Set}$ is also $\textbf{Cppo}$-enriched \cite{HasJacSok}. Order on hom-sets is imposed by the point-wise order and for any $X$  the set $\mathcal{M}_1 X=X +  1=X +  \{\perp\}$ is a poset whose partial order $\leqslant$ is given by $x\leqslant y$ iff $x=\perp$ or $x=y$. 
\end{example}

\subsubsection{Monads on Kleisli categories} 
In this paper we will often work with monads on Kleisli categories. Here we list basic properties of such monads. Everything presented below with the exception of the last theorem follows easily by classical results in category theory (see e.g. \cite{MacLane}). Assume that $(T,\mu,\eta)$ is a monad on $\mathsf{C}$ and $S:\mathsf{C}\to\mathsf{C}$ is a functor that lifts to $\overline{S}:\mathcal{K}l(T)\to \mathcal{K}l(T)$ with the associated distributive law $\lambda:ST\implies TS$.
Moreover, let $(\overline{S},m,e)$ be a monad on $\mathcal{K}l(T)$. We have the following two adjoint situations whose composition is an adjoint situation \cite{MacLane}.
$$
\xymatrix@-1pc{
\mathsf{C}\ar@/^1pc/[r]^{^{\sharp}}\ar@{}[r]|\perp & \mathcal{K}l(T)\ar@/^1pc/[r]^{^{\sharp}}\ar@/^1pc/[l]^{U_T}\ar@{}[r]|\perp & \mathcal{K}l(\overline{S})\ar@/^1pc/[l]^{U_{\overline{S}}}
}
$$
This yields a monadic structure on the functor $TS:\mathsf{C}\to \mathsf{C}$. The $X$-components of the multiplication $\mathfrak{m}$ and the unit $\mathfrak{e}$ of the monad $TS$ are given by:
$$
\mathfrak{m}_X = \mu_{SX}\circ T\mu_{SX} \circ TT m_X \circ T\lambda_{SX} \quad \text{and}\quad \mathfrak{e}_X = e_X.
$$
The composition $\cdot$ in $\mathcal{K}l(TS) = \mathcal{K}l(\overline{S})$ is given in terms of the composition in $\mathsf{C}$ as follows. For $f:X\to TSY$ and $g:Y\to TSZ$ we have:
$$
\xymatrix{
X\ar@{-->}[d]^{g\cdot f} \ar[r]^f & TSY \ar[r]^{TSg} &
TSTSZ \ar[r]^{T\lambda_{SZ}} & TTSSZ \ar[d]^{TT(m_Z)} \\
 TSZ & &TTSZ\ar[ll]^{\mu_{SZ}} & TTTSZ\ar[l]^{T\mu_{SZ}} 
}
$$

The following result can be proved by straightforward verification.
\begin{theorem}\label{theorem_Cppo_continuous}
Assume that $\mathcal{K}l(T)$ is $\textbf{Cppo}$-enriched and  $\overline{S}$ is locally continuous. Then $\mathcal{K}l(TS)=\mathcal{K}l(\overline{S})$ is $\textbf{Cppo}$-enriched.
\end{theorem}

\section{Hiding internal moves inside a monadic structure}
\label{subsection:monads}

Throughout this paper we assume that $(T,\mu,\eta)$ is a monad on a category $\mathsf{C}$ with binary coproducts. Let $ + $ denote the binary coproduct operator in $\mathsf{C}$. Assume that $F:\mathsf{C}\to\mathsf{C}$ is a functor and let $F_\tau=F +  \mathcal{I}d$. In this paper we deal with functors of the form $TF_\tau = T(F +  \mathcal{I}d)$. 
Labelled transition system and $\varepsilon$-NA functor are of this form since 
\begin{align*}
&\mathcal{P}(\Sigma_\tau \times \mathcal{I}d)\cong \mathcal{P}(\Sigma\times \mathcal{I}d +  \mathcal{I}d)=\mathcal{P}(F +  \mathcal{I}d) \text{ for } F=\Sigma\times \mathcal{I}d \text{ and }\\
& \mathcal{P}(\Sigma_\varepsilon \times \mathcal{I}d +  1)\cong \mathcal{P}(\Sigma\times \mathcal{I}d +  1  +  \mathcal{I}d) = \mathcal{P}(F +  \mathcal{I}d) \text{ for }F=\Sigma\times \mathcal{I}d +  1.
\end{align*}
The functor $F$ represents the visible part of the structure, whereas the functor $\mathcal{I}d$ represents silent moves. Functors of this type were used to consider $\varepsilon$-elimination from coalgebraic perspective in \cite{HasJacSok_jap,SilWester}.
In \cite{Brengos13} we noticed that given some mild assumptions on the monad $T$, the functor $TF_\tau$ can itself be turned into a monad or embedded into one. The aim of this section is to recall these results here. Before we do it, we will list basic definitions and properties concerning  categories and monads used in the construction.    

\subsubsection{Basic definitions and properties}

  For a family of objects $\{X_k\}_{k\in I}$ if the coproduct $\coprod_i X_i$ exists then by $\iota^k:X_k \to \coprod_k X_k$ we denote the coprojection into $k$-th component of $\coprod_k X_k$.

We say that a category is \emph{a category with zero morphisms} if for any two objects $X,Y$ there is a morphism $0_{X,Y}$ which is an annihilator w.r.t. composition. To be more precise $f\circ 0 = 0=0\circ g$ for any morphisms $f,g$ with suitable domain and codomain. 
\begin{example}
For the monad $T\in \{\mathcal{P},\mathcal{M}_1\}$ on $\mathsf{Set}$ the category $\mathcal{K}l(T)$ is a category with zero morphisms given by $\perp:X\to \mathcal{P}Y; x\mapsto \varnothing$ for $\mathcal{P}$ and $\perp:X\to \mathcal{M}_1 Y; x\mapsto \perp$ for the monad $\mathcal{M}_1$.
\end{example}  
Given two monads $(S,\mu^S,\eta^S)$ and $(S',\mu^{S'},\eta^{S'})$ a \emph{monad morphism} $h$ is a natural transformation $h:S\implies S'$ which preserves unit and multiplication of monads, i.e. $h\circ \eta^{S}  =\eta^{S'}$ and $h\circ \mu^{S} = \mu^{S'}\circ hh$. A \emph{free monad} over a functor $F:\mathsf{C}\to \mathsf{C}$ \cite{Barr,Manes} is a monad $(F^{*},m,e)$ together with a natural transformation $\nu:F\implies F^{*}$ such that for any monad $(S,m^S,e^S)$ on $\mathsf{C}$ and a natural transformation $s:F\implies S$ 	there is a unique monad morphism $h:(F^{*},m,e)\to (S,m^S,e^S)$ such that the following diagram commutes:
$$
\xymatrix@-0.5pc{
F\ar@{=>}[r]^{\nu}\ar@{=>}[dr]_{s} & F^{*} \ar@{=>}[d]^{h} \\
& S
}
$$ 

\begin{theorem}\cite{Barr}\label{theorem:free_functor}
Assume that for an endofunctor $F:\mathsf{C}\to \mathsf{C}$ and any object $X$ the free $F$-algebra over $X$ (=initial $F(-) +  X$-algebra) $i_X$ exists in $\mathsf{C}^F$. For an object $X$ and a morphism $f:X\to Y$ in $\mathsf{C}$ let $F^{*}X$ denote the carrier of $i_X$ and $F^{*}f:F^{*}X\to F^{*}Y$ denote the unique morphism for which the following diagram commutes:
$$
\xymatrix@-0.5pc{
FF^{*}X+X\ar[rr]^{i_X}\ar@{-->}[d]_{F(F^{*}f)+id_X} & & F^{*}X\ar@{-->}[d]^{F^{*}f} \\
FF^{*}Y+X\ar[r]_{id+f}& 
FF^{*}Y+Y\ar[r]_{i_Y} & F^{*}Y
}
$$  
The assignment $F^{*}$ is functorial and 
can be naturally equipped with a monadic structure $(F^{*},m,e)$ which is a consequence of the universal properties of $i_X$. Moreover, this monad is the free monad over $F$.
\end{theorem}
In the sequel we assume the following:

\begin{itemize}
\item The functor $F:\mathsf{C}\to \mathsf{C}$ lifts to $\overline{F}:\mathcal{K}l(T)\to\mathcal{K}l(T)$. As a direct consequence we get that $F_\tau = F +  \mathcal{I}d$ lifts to a functor $\overline{F_\tau}=\overline{F} +  \mathcal{I}d$ on $\mathcal{K}l(T)$. This follows by the fact that coproducts in $\mathcal{K}l(T)$ come from coproducts in the base category (see also e.g. \cite{HasJacSok} for a discussion on liftings of coproducts of functors).
\item The functor $F$ admits the free $F$-algebra $i_X$ in $\mathsf{C}^F$ for any object $X$. By theorem above this yields the free monad $(F^{*},m,e)$ over $F$ in $\mathsf{C}$.
\end{itemize}

\subsubsection{Monadic structure on $TF_\tau$}
The aim of this subsection is to present the first strategy towards handling the invisible part of computation by a monadic structure. 
  Note that in the following result all morphisms, in particular all coprojections and mediating morphisms, live in $\mathcal{K}l(T)$.
\begin{theorem}\cite{Brengos13}\label{theorem:monad_on_TF} If $\mathcal{K}l(T)$ is a category with zero morphisms then
the triple $(\overline{F_\tau},m',e')$, where $e':\mathcal{I}d\multimap \overline{F} +  \mathcal{I}d; e_X' = \iota^2$
and 
\begin{align*}
&m':\overline{F}(\overline{F} +  \mathcal{I}d) +  (\overline{F} +  \mathcal{I}d)\stackrel{\overline{F}([0,id]) +  id}{\multimap} \overline{F}  +   (\overline{F} +  \mathcal{I}d) \stackrel{[\iota^1,id]}{\multimap} \overline{F} +  \mathcal{I}d
\end{align*}
is a monad on $\mathcal{K}l(T)$. Two adjoint situations $\mathsf{C}\rightleftarrows \mathcal{K}l(T) \rightleftarrows \mathcal{K}l(\overline{F_\tau})$ yield a monadic structure on $TF_\tau$.
\end{theorem}
The composition $\cdot$ in $\mathcal{K}l(TF_\tau)=\mathcal{K}l(\overline{F_\tau})$ is given as follows. Let $\lambda:F_\tau T\implies TF_\tau$ denote the distributive law associated with the lifting $\overline{F_\tau}$ of $F_\tau$. For any  $f:X\to TF_\tau Y$, $g:Y\to TF_\tau Z$ we have:
$$
g\cdot f = \mu_{F_\tau Z}\circ T\mu_{F_\tau Z} \circ TTm'_Z \circ T \lambda_{F_\tau Z} \circ TF_\tau g \circ f. 
$$

We illustrate the above construction in the following example, where $T=\mathcal{P}$ and $F_\tau = \Sigma_\tau \times \mathcal{I}d$.
\begin{example} \label{example:LTS_monad} As mentioned before, the monad $\mathcal{P}$ (as any other monad on $\mathsf{Set}$) comes with strength $\text{st}$ which lifts the functor $\Sigma_\tau \times \mathcal{I}d:\mathsf{Set}\to\mathsf{Set}$ to the functor $\overline{\Sigma_\tau}:\mathcal{K}l(\mathcal{P})\to \mathcal{K}l(\mathcal{P})$.
For the functor $\overline{\Sigma_\tau}\cong \overline{\Sigma}+\mathcal{I}d$ we define the multiplication $m'$ and the unit $e'$ as in Theorem \ref{theorem:monad_on_TF}. For any set $X\in \mathcal{K}l(\mathcal{P})$ we put $m_X':\overline{\Sigma_\tau}\overline{\Sigma_\tau} X\multimap \overline{\Sigma_\tau} X$  and $e_X':X\multimap \overline{\Sigma_\tau}X$ to be:
$$
m_X'(\sigma_1,\sigma_2,x)  = \left\{ \begin{array}{cc} \{(\sigma_1,x)\} & \text{if } \sigma_2 = \tau, \\
\{(\sigma_2,x)\} & \text{if } \sigma_1 = \tau, \\
\varnothing & \text{otherwise}\end{array}\right.
\qquad e_X'(x) = \{(\tau,x)\}.
$$
By Theorem \ref{theorem:monad_on_TF} the triple $(\overline{\Sigma_\tau},m',e')$ is a monad on $\mathcal{K}l(\mathcal{P})$. By composing the two adjoint situations we get a monadic structure on the LTS functor.
The composition in $\mathcal{K}l(\mathcal{P}(\Sigma_\tau \times~\mathcal{I}d))$ is given as follows. For $f:X\to \mathcal{P}(\Sigma_\tau \times Y)$ and $g:Y\to \mathcal{P}(\Sigma_\tau \times Z)$ we have $g\cdot f:X\to \mathcal{P}(\Sigma_\tau \times Z)$:
$$
g\cdot f (x) = \{(\sigma,z)\mid x\stackrel{\sigma}{\to}_f y \stackrel{\tau}{\to}_g z \text{ or }x\stackrel{\tau}{\to}_f y \stackrel{\sigma}{\to}_g z \text{ for some }y\in Y\}.
$$
\end{example}

The construction provided by Theorem \ref{theorem:monad_on_TF} can be applied 
only when $\mathcal{K}l(T)$ is a category with zero morphisms. Some monads fail to have this property. For example, if instead of considering the monad $\mathcal{P}$ we consider the non-empty powerset monad $\mathcal{P}_{\neq \varnothing}$. In what follows we focus on the second strategy for handling internal transitions by a monadic structure on the functor which does not require from $\mathcal{K}l(T)$ to be a category with zero morphisms. 

\subsubsection{Monadic structure on $TF^{*}$} Here, we present an approach towards dealing with silent moves which uses free monads. At the beginning of this section we stated clearly that the coalgebras we are dealing with are of the type $TF_\tau$. Any $TF_\tau$-coalgebra $\alpha:X\to TF_\tau X$ can be turned into a $TF^{*}$-coalgebra $\underline{\alpha}:X\to TF^{*}X$ by putting
$$\underline{\alpha} = T([\nu_X,e_X])\circ \alpha,$$
where the mono-transformation $[\nu,e]:F_\tau\implies F^{*}$ comes from the definition of a free monad.

\begin{example}
Consider the LTS functor $\mathcal{P}(\Sigma_\tau \times \mathcal{I}d) \cong \mathcal{P}(\Sigma\times \mathcal{I}d  +  \mathcal{I}d)$ and let $F=\Sigma\times \mathcal{I}d$. The free monad over $F$ in $\mathsf{Set}$ is given by $(\Sigma^{*}\times \mathcal{I}d,m,e)$, where $\Sigma^{*}$ is the set of finite words over $\Sigma$ together with the empty string $\varepsilon\in \Sigma^{*}$ and $m$ and $e$ are given for any set $X$ as follows:
\begin{align*}
& m_X:\Sigma^{*}\times \Sigma^{*}\times X\to \Sigma^{*}\times X; (s,s',x) \mapsto (ss',x) \text{ and }\\
& e_X: X\to \Sigma^{*}\times X; x\mapsto (\varepsilon,x).
\end{align*}
For any $\alpha:X\to \mathcal{P}(\Sigma_\tau\times X)$ we define $\underline{\alpha}:X\to \mathcal{P}(\Sigma^{*}\times X)$ by $$\underline{\alpha}(x) = \{(a,y) \mid (a,y)\in \alpha(x) \text{ and } a\in \Sigma\}\cup \{(\varepsilon,y) \mid (\tau,y)\in \alpha(x)\}.$$ 
\end{example}

\begin{example}\label{example_NA_translate}
The $\varepsilon$-NA's  are coalgebras of the type $TF_\tau$ for the monad $T=\mathcal{P}$ and $F=\Sigma\times \mathcal{I}d +  1$. The functor $F=\Sigma\times\mathcal{I}d +  1$ lifts to $\mathcal{K}l(\mathcal{P})$ \cite{HasJacSok} and admits all free $F$-algebras. Let $F^{*}$ denote the free monad over $F$. The functor $F^{*}:\mathsf{Set}\to\mathsf{Set}$ is defined on objects and morphisms by
\begin{align*}
&F^{*}X= \Sigma^{*}\times X +  \Sigma^{*},\\
&F^{*}f: \Sigma^{*}\times X +  \Sigma^{*}\to \Sigma^{*}\times Y +  \Sigma^{*}; F^{*}f = (id_{\Sigma^{*}} \times f) +  id_{\Sigma^{*}} \text{ for }f:X\to Y.
\end{align*}
The monadic structure $(F^{*},m,e)$  is given by:
\begin{align*}
&m_X:\Sigma^{*}\times (\Sigma^{*}\times X + \Sigma^{*})+ \Sigma^{*}\to \Sigma^{*}\times X+ \Sigma^{*}; \\
&m_X(s_1,s_2,x)=(s_1s_2,x) \quad m_X(s_1,s_2) = s_1s_2 \quad m_X(s_1)=s_1,\\
&e_X:X\to \Sigma^{*}\times X+ \Sigma^{*}; x\mapsto (\varepsilon,x).
\end{align*}
For any $\varepsilon$-NA coalgebra $\alpha:X\to \mathcal{P}(\Sigma_\varepsilon \times X+1)$ we define $$\underline{\alpha}:X\to \mathcal{P}(\Sigma^{*} \times X+\Sigma^{*}); x\mapsto \{(a,y)\in \Sigma^{*}\times X\mid (a,y)\in \alpha(x)\}\cup A_x, $$
where $A_x= \textbf{if } \checked \in \alpha(x) \textbf{ then }\{\varepsilon\} \textbf{ else }\varnothing$.
\end{example}

In order to proceed with the construction we need one additional lemma.

\begin{lemma}\cite{Brengos13}
The algebra $i_X^\sharp=
\eta_{F^{*}X}\circ i_X:FF^{*}X+X\to TF^{*}X$ is the free $\overline{F}$-algebra over $X$ in $\mathcal{K}l(T)^{\overline{F}}$.
\end{lemma}
 Let $\overline{F}^{*}:\mathcal{K}l(T)\to\mathcal{K}l(T)$ be the functor obtained by following the guidelines of Theorem \ref{theorem:free_functor} using the family $\{i^\sharp_X\}_{X\in \mathcal{K}l(T)}$ of free algebras in $\mathcal{K}l(T)^{\overline{F}}$. 
\begin{theorem}\cite{Brengos13}\label{theorem:monad_free}
We have the following:
\begin{enumerate}
\item $F^{*}:\mathsf{C}\to\mathsf{C}$ lifts to $\overline{F}^{*}:\mathcal{K}l(T)\to \mathcal{K}l(T)$,
\item  $(\overline{F}^{*},m^\sharp,e^\sharp)$ is the free monad over $\overline{F}$ in $\mathcal{K}l(T)$.
\end{enumerate}
Two adjoint situations $\mathsf{C}\rightleftarrows \mathcal{K}l(T)\rightleftarrows \mathcal{K}l(\overline{F}^{*})$ yield a monadic structure on $TF^{*}$.
\end{theorem}
The composition $\cdot$ in $\mathcal{K}l(TF^{*})=\mathcal{K}l(\overline{F}^{*})$ is given as follows. Let $\lambda:F^{*}T\implies TF^{*}$ denote the distributive law associated with the lifting $\overline{F}^{*}$ of $F^{*}$. The composition of  $f:X\to TF^{*}Y$, $g:Y\to TF^{*}Z$ in $\mathcal{K}l(TF^{*})$ is given by:
\begin{align*}
&g\cdot f = \mu_{F^{*}Z} \circ T\mu_{F^{*}Z}\circ  TTm_Z^\sharp \circ T \lambda_{F^{*}Z} \circ TF^{*} g \circ f=\\
&\mu_{F^{*}Z} \circ  T\mu_{F^{*}Z}\circ TT(\eta_Z\circ m_Z) \circ T \lambda_{F^{*}Z} \circ TF^{*} g \circ f = \\
&\mu_{F^{*}Z} \circ TTm_Z \circ T \lambda_{F^{*}Z} \circ TF^{*} g \circ f.
\end{align*}

\begin{example} \label{example:LTS_monad_extended}
The composition $\cdot$ in $\mathcal{K}l(\mathcal{P}(\Sigma^{*}\times \mathcal{I}d))$ is given by the following formula. For $f:X\to \mathcal{P}(\Sigma^{*}\times Y)$ and $g:Y\to \mathcal{P}(\Sigma^{*}\times Z)$ we have $g \cdot f:X\to \mathcal{P}(\Sigma^{*}\times Z)$:
$$
g\cdot f (x) = \{(s_1s_2,z)\mid x\stackrel{s_1}{\to}_f y \stackrel{s_2}{\to}_g z \text{ for some }y\in Y \text{ and }s_1,s_2\in \Sigma^{*}\}.
$$
We call the monad $\mathcal{P}(\Sigma^{*}\times \mathcal{I}d)$  \emph{free LTS monad}.
\end{example}

\begin{example}\label{example:na_composition}
The composition $\cdot$ in $\mathcal{K}l(\mathcal{P}(\Sigma^{*}\times \mathcal{I}d +  \Sigma^{*}))$ is given by the following formula. For $f:X\to \mathcal{P}(\Sigma^{*}\times Y +  \Sigma^{*})$ and $g:Y\to \mathcal{P}(\Sigma^{*}\times Z +  \Sigma^{*})$ we have $g \cdot f:X\to \mathcal{P}(\Sigma^{*}\times Z +  \Sigma^{*})$:
\begin{align*}
g\cdot f (x) =& \{(s_1s_2,z)\mid x\stackrel{s_1}{\to}_f y \stackrel{s_2}{\to}_g z \text{ for some }y\in Y \text{ and }s_1,s_2\in \Sigma^{*}\}\cup\\
 &\{s_1 s_2 \mid x\stackrel{s_1}{\to}_f y  \text{ and } s_2 \in g(y) \text{ for some }y\in Y\}\cup \{s_1 \mid s_1\in f(x) \}.
 \end{align*}
 We call $\mathcal{P}(\Sigma^{*}\times \mathcal{I}d +  \Sigma^{*})$ monad  \emph{free $\varepsilon$-NA monad} or \emph{$\varepsilon$-NA monad} in short.
\end{example}

We see that if we deal with functors of the form $T(F + \mathcal{I}d)$, where $T$ is a monad, given some mild assumptions on $T$ and $F$ we may deal with the silent and observable part of computation inside a monadic structure on the functor $TF_\tau$ itself or by embedding the functor $TF_\tau$ into the monad $TF^{*}$ by the natural transformation ${F_\tau}\implies {F}^{*}$. Therefore, from now on the term ``coalgebras with internal moves'' becomes synonymous to ``coalgebras over a monadic type''. Weak bisimulation and, as we will also see, trace equivalence are defined for coalgebras over monadic types, without the need for specifying visible and silent part of the structure.

\section{Weak bisimulation}  \label{section:weak_bisimulation}
\label{section:LTS_coalgebraically} 
In this section we recall classical definition(s) of weak bisimulation for labelled transition systems and coalgebraic constructions from \cite{Brengos13}. Weak bisimulation for labelled transition systems can be defined as a strong bisimulation on a saturated structure. Process of saturation can be described as taking the reflexive and transitive closure of a given structure w.r.t. the suitable composition and order. First of all we present a paragraph devoted to classical definitions of weak bisimulation for LTS. Then we show how Kleisli compositions from Examples \ref{example:LTS_monad} and \ref{example:LTS_monad_extended} play role  in the LTS saturation. These examples motivate the definition of an order saturation monad and weak bisimulation \cite{Brengos13}. What is essentially new in this section is the following. First of all we present a definition of weak bisimulation in terms of a kernel bisimulation on the saturated structure and not via lax- and oplax-homomorphisms in Aczel-Mendler style as it was done in \cite{Brengos13}. Second of all, the last paragraph compares the two generalizations of the strategies towards saturation from the point of view of weak bisimulation which was not done in \cite{Brengos13}.

\subsubsection{Weak bisimulation for LTS} Let $\alpha:X\to \mathcal{P}(\Sigma_\tau\times X)$ be a labelled transition system coalgebra. For $\sigma\in \Sigma_\tau$ and $s\in \Sigma^{*}$ define the relations  $\stackrel{\sigma}{\implies},\stackrel{s}{\to},\stackrel{s}{\implies}\subseteq X\times X$  by
\begin{align*}
\stackrel{\sigma}{\implies} &= \left \{\begin{array}{cc} (\stackrel{\tau}{\to})^{*}& \text{ if }\sigma = \tau \\
(\stackrel{\tau}{\to})^{*}\circ \stackrel{\sigma}{\to}\circ (\stackrel{\tau}{\to})^{*} & \text{ otherwise,}\end{array}\right. 
\quad \stackrel{s}{\to}  =\left \{\begin{array}{cc} \stackrel{\tau}{\to}& \text{ if }s=\varepsilon \\
\stackrel{\sigma_1}{\to}\circ \ldots \circ \stackrel{\sigma_n}{\to} & \text{for }s=\sigma_1\ldots \sigma_n,
\end{array}\right.\\
\stackrel{s}{\implies}  &=\left \{\begin{array}{cc} (\stackrel{\tau}{\to})^{*}& \text{ if }s \text{ is the empty word} \\
(\stackrel{\tau}{\to})^{*}\circ \stackrel{\sigma_1}{\to}\circ (\stackrel{\tau}{\to})^{*}\circ \ldots \circ (\stackrel{\tau}{\to})^{*}\circ \stackrel{\sigma_n}{\to}\circ (\stackrel{\tau}{\to})^{*}  & \text{for }s=\sigma_1\ldots \sigma_n\end{array}\right.
\end{align*} 
where, given any relation $R\subseteq X\times X$, the symbol $R^{*}$ denotes the reflexive and transitive closure of $R$.  We now present four different but equivalent definitions of weak bisimulation for LTS's. Due to limited space we do so in one definition block. 
\begin{definition}\cite{Milner,Milner3,Sangiorgi11}
\label{definition:LTS_weak_bisimulation}
 A relation $R\subseteq X\times X$ is called  \emph{weak bisimulation} on $\alpha$ if the following condition holds. If $(x,y)\in R$ then 
\begin{enumerate}[(i)]
\item for any $\sigma\in \Sigma_\tau$ the condition $x\stackrel{\sigma}{\to}x' \text{ implies } y\stackrel{\sigma}{\implies }y'$ \label{def:1}
\item for any $\sigma\in \Sigma_\tau$ the condition $x\stackrel{\sigma}{\implies }x' \text{ implies } y\stackrel{\sigma}{\implies }y'$ \label{def:2}
\item for any $s\in \Sigma^{*}$  the condition $x\stackrel{s}{\to}x' \text{ implies } y\stackrel{s}{\implies }y'$ \label{def:3}
\item for any $s\in \Sigma^{*}$  the condition $x\stackrel{s}{\implies}x' \text{ implies } y\stackrel{s}{\implies }y'$ \label{def:4}
\end{enumerate}
and $y'\in X$ such that $(x',y')\in R$ and a symmetric statement holds. 
\end{definition}

In this paper we will focus on Definitions \ref{definition:LTS_weak_bisimulation}.\ref{def:2} and \ref{definition:LTS_weak_bisimulation}.\ref{def:4} and their generalization. They both suggest that weak bisimulation can be defined as a strong bisimulation on a \emph{saturated model}. It is worth noting that in our previous paper we focused on analogues of Def \ref{definition:LTS_weak_bisimulation}.\ref{def:1} and \ref{definition:LTS_weak_bisimulation}.\ref{def:3} and comparison with Def \ref{definition:LTS_weak_bisimulation}.\ref{def:2} and \ref{definition:LTS_weak_bisimulation}.\ref{def:4} respectively (see \cite{Brengos13} for details). 

\subsubsection{Saturation for LTS coalgebraically}

Let us assume that $\cdot$ is a composition in $\mathcal{K}l(\mathcal{P}(\Sigma_\tau\times \mathcal{I}d))$ as in Example \ref{example:LTS_monad}. Given an LTS coalgebra $\alpha:X\to \mathcal{P}(\Sigma_\tau\times X)$ the saturated LTS $\alpha^{*}:X\to \mathcal{P}(\Sigma_\tau\times X)$ is obtained as follows:
$
\alpha^{*} = 1_X  \vee  \alpha  \vee  \alpha\cdot \alpha  \vee \ldots = \bigvee_{n=0,1,2\ldots} \alpha^n$,
where $\bigvee$ denotes supremum in the complete lattice $(\mathcal{P}(\Sigma_\tau\times X)^X,\leqslant)$, where the relation $\leqslant$ is given by $\alpha\leqslant \beta \iff \alpha(x)\subseteq \beta(x) \text{ for any }x\in X$.
We see that for $(\sigma,y)\in \Sigma_\tau\times X$:
$
(\sigma,y)\in \alpha^{*}(x) \text{ if and only if }x\stackrel{\sigma}{\implies}_\alpha y$. Weak bisimulation on $\alpha$ according to Definition \ref{definition:LTS_weak_bisimulation}.\ref{def:2}  is a strong bisimulation on $\alpha^{*}$.

If we now consider $\cdot $ to be composition in $\mathcal{K}l(\mathcal{P}(\Sigma^{*}\times \mathcal{I}d))$ as in Example \ref{example:LTS_monad_extended} for an LTS considered as a $\mathcal{P}(\Sigma^{*}\times \mathcal{I}d)$-coalgebra $\alpha:X\to \mathcal{P}(\Sigma^{*}\times X)$ define $\alpha^{*}:X\to \mathcal{P}(\Sigma^{*}\times X)$ to be
$
\alpha^{*} = 1_X  \vee  \alpha  \vee  \alpha\cdot \alpha  \vee \ldots = \bigvee_{n=0,1,2\ldots} \alpha^n$. Then $(s,y)\in \alpha^{*}(x) \text{ if and only if } x\stackrel{s}{\implies}_\alpha y$ for any $s\in \Sigma^{*}$.
Weak bisimulation from Def. \ref{definition:LTS_weak_bisimulation}.\ref{def:4} is a strong bisimulation on $\alpha^{*}$.

\subsubsection{Saturation for $T$-coalgebras} 
\label{subsection:saturation}
A monad $T$ whose Kleisli category is order-enriched is called \emph{ ordered $*$-monad }  or \emph{ ordered saturation monad } \cite{Brengos13} provided that in $\mathcal{K}l(T)$ for any morphism $\alpha: X \multimap X$ there is a morphism $\alpha^{*}:X\multimap X$ satisfying the following conditions:
\begin{enumerate}[(a)]
\item \label{axiom:1} $1\leqslant \alpha^{*}$,
\item \label{axiom:2} $\alpha\leqslant \alpha^{*}$,
\item \label{axioms:3} $\alpha^{*}\cdot \alpha^{*}\leqslant \alpha^{*}$,
\item \label{axiom:3.5} if $\beta:X\multimap X$ satisfies $1\leqslant \beta$, $\alpha\leqslant \beta$ and $\beta\cdot \beta\leqslant \beta$ then $\alpha^{*}\leqslant \beta$,
\item \label{axioms:4} for any $f:X\to Y$ in $\mathsf{C}$ and any $\beta:Y\multimap Y$ in $\mathcal{K}l(T)$ we have:
$$
f^{\sharp}\cdot \alpha \mathrel{\Box} \beta \cdot f^{\sharp} \implies f^{\sharp}\cdot \alpha^{*} \mathrel{\Box} \beta^{*} \cdot f^{\sharp} \text{ for } \mathrel{\Box} \in \{\leqslant,\geqslant\}.
$$
\end{enumerate}
For the rest of the section we assume that $T$ is an order saturation monad with the saturator operator $(-)^{*}$.

\begin{remark}
We could try and define $\alpha^{*}$ as the least fix point $\mu x.(1 \vee x\cdot \alpha)$. Indeed, if $T$ is e.g. complete join-semilatice enriched monad then the saturated structure is defined this way. We believe that our definition is slightly more general as it does not require for the mapping $x\mapsto 1 \vee x\cdot \alpha$ to be well defined. Intuitively however,
$\alpha^{*}$ should and will be associated with $\mu x.(1 \vee x\cdot \alpha)$.
\end{remark}

\begin{example}
The powerset monad $\mathcal{P}$ and the non-empty powerset monad $\mathcal{P}_{\neq \varnothing}$ are examples of order saturation monads \cite{Brengos13}.  The monads from Examples \ref{example:LTS_monad} and \ref{example:LTS_monad_extended} are order saturation monads \cite{Brengos13}. Also the $\mathcal{CM}$ monad of convex distributions described in \cite{Jacobs08} is an order saturation monad \cite{Brengos13}. Although we will not focus on $\mathcal{CM}$ in this paper it is a very important monad that  is used to model Segala systems, their trace semantics and  probabilistic weak bisimulations \cite{Brengos13,Jacobs08,Segala,SegalaThesis}. Any Kleene monad \cite{Gov} is also an order saturation monad \cite{Brengos13}.
\end{example}

Since $\mathcal{P}$, $\mathcal{P}(\Sigma_\tau\times \mathcal{I}d)$ and $\mathcal{P}(\Sigma^{*}\times \mathcal{I}d)$ are order saturation monads, the following question arises: is the saturation operator for LTS monads related to saturation in $\mathcal{K}l(\mathcal{P})$? The following theorem answers that question in general and shows the relation between a saturation operator in $\mathcal{K}l(T)$ and $\mathcal{K}l(TS)$ for a monad $\overline{S}$ on $\mathcal{K}l(T)$.
\begin{theorem}\label{theorem:S_star_monad} \cite{Brengos13}
Assume $S:\mathsf{C}\to \mathsf{C}$ lifts to $\overline{S}:\mathcal{K}l(T)\to\mathcal{K}l(T)$ and $(\overline{S},m,e)$ is a monad on $\mathcal{K}l(T)$. If  $\overline{S}$ is locally monotonic and satisfies the equation $$m_X\cdot \overline{S}[(m_X\cdot \overline{S}\alpha)^{*}\cdot e_X] = (m_X\cdot \overline{S}\alpha)^{*}$$ for any $\alpha:X\multimap  \overline{S}X$, then the monad $TS$ is an order saturation monad with the saturation operator $(-)^{\bigstar}$  given  by
$\alpha^{\bigstar} = (m_X\cdot \overline{S}\alpha)^{*}\cdot e_X$.
\end{theorem}
If $T=\mathcal{P}$ and $S$ is taken either to be $\Sigma_\tau\times \mathcal{I}d$ or $\Sigma^{*}\times \mathcal{I}d$, then the lifting $\overline{S}$ exists and is equipped with a monadic structure as in Section \ref{subsection:monads}. Moreover, $\overline{S}$ satisfies the assumptions of Theorem \ref{theorem:S_star_monad} \cite{Brengos13}. In other words, the LTS saturations for $\mathcal{P}(\Sigma_\tau\times \mathcal{I}d)$ and $\mathcal{P}(\Sigma^{*}\times \mathcal{I}d)$ are obtained respectively by
$$
(m_X'\cdot \overline{\Sigma_\tau}\alpha)^{*}\cdot e_X' \text{ and } (m_X^\sharp\cdot \overline{\Sigma}^{*}\alpha)^{*}\cdot e_X^\sharp.
$$
In sections to come we will deal with generalizations of these two saturations and check under which conditions they yield the same notion of weak bisimulation.

\subsubsection{Weak bisimulation for $T$-coalgebras}
The following slogan should be in our opinion considered the starting point to the theory of weak bisimulation for $T$-coalgebras: \emph{weak bisimulation on $\alpha:X\to TX$ =  bisimulation on $\alpha^{*}:X\to TX$}.

\begin{definition}
Let $\alpha:X\to TX$ be a $T$-coalgebra. A relation $X\stackrel{\pi_1}{\leftarrow} R \stackrel{\pi_2}{\to} X$ is \emph{weak bisimulation} on $\alpha$ if it is a bisimulation on $\alpha^{*}$.
\end{definition}

We see that the above definition coincides with the standard definition of weak bisimulation for LTS considered as $\mathcal{P}(\Sigma_\tau\times \mathcal{I}d)$- and $\mathcal{P}(\Sigma^{*}\times \mathcal{I}d)$-coalgebras.

\subsubsection{Weak bisimulation for $TF_\tau$- and $TF^{*}$-coalgebras}\label{subsection:weak_bisim_for_TF_and_TF_star}
This subsection will be devoted to comparing both approaches towards defining weak bisimulation for $TF_\tau$-coalgebras that generalize Def. \ref{definition:LTS_weak_bisimulation}.\ref{def:2} and \ref{definition:LTS_weak_bisimulation}.\ref{def:4} for LTS. Here, we additionally assume that $\mathcal{K}l(T)$ is a category with zero morphisms. Then we may either define a monadic structure on $TF_\tau$ or embed the functor into the monad $TF^{*}$. These two approaches applied for LTS give two different saturations, yet the weak bisimulations coincide. It is natural to suspect that given some mild assumptions it will also be the case in a more general setting. We will now list all the necessary ingredients.  

We assume  $(\overline{F_\tau},m',e')$ and $(\overline{F}^{*},m^\sharp,e^\sharp)$ are monads as in Section \ref{subsection:monads} and that both satisfy the assumptions of Theorem~\ref{theorem:S_star_monad} for the monad $\overline{S}$. For sake of simplicity and clarity of notation we will drop $^\sharp$ and write $(\overline{F}^{*},m,e)$ instead of $(\overline{F}^{*},m^\sharp,e^\sharp)$. The consequences of these assumptions are the following:
\begin{itemize}
  \item A natural transformation $\nu:\overline{F}\implies \overline{F}^{*}$ which arises by the definition of a free monad.
  \item A natural transformation $\iota^1:\overline{F}\implies \overline{F_\tau} = \overline{F +  \mathcal{I}d}=\overline{F} +  \mathcal{I}d$. This transformation is given regardless of the assumptions.
  \item Unique monad morphism 
  $h:(\overline{F}^{*},m,e)\multimap (\overline{F_\tau},m',e')$ in $\mathcal{K}l(T)$ making the first three diagrams commute:
$$
\xymatrix@-1pc{
\overline{F} \ar@{-o}[r]^{\nu}\ar@{-o}[dr]_{\iota^1} & \overline{F}^{*}\ar@{-o}[d]^h\\
& \overline{F_\tau}
}\qquad 
\xymatrix@-1pc{
\mathcal{I}d \ar@{-o}[r]^{e}\ar@{-o}[dr]_{e'=\iota^2} & \overline{F}^{*}\ar@{-o}[d]^h\\
& \overline{F_\tau}
}\qquad 
\xymatrix@-1pc{
\overline{F}^{*}\overline{F}^{*}\ar@{-o}[r]^{m}\ar@{-o}[d]_{hh} & \overline{F}^{*}\ar@{-o}[d]^h\\
\overline{F_\tau}\overline{F_\tau}\ar@{-o}[r]^{m'} & \overline{F_\tau}
}\qquad 
\xymatrix@-1pc{
\overline{F_\tau} \ar@{-o}[r]^{[\nu,e]}\ar@{-o}[dr]_{[\iota^1,e']=id} & \overline{F}^{*}\ar@{-o}[d]^h\\
& \overline{F_\tau}
}
$$  
 Commutativity of the first two diagrams implies commutativity of the forth.  Existence and uniqueness of $h$ follows by the fact that $\overline{F}^{*}$ is a free monad over $\overline{F}$ and $\iota^1:\overline{F}\implies \overline{F_\tau}$ is a natural transformation.
 
 \item The monads $TF_\tau$ and $TF^{*}$ are order saturation monads. The saturation operators $(-)^\bigstar$ and $(-)^\divideontimes$ for $TF_\tau$- and $TF^{*}$-coalgebras resp. are given as follows. Let $\alpha:X\multimap \overline{F_\tau}X$ (i.e. $\alpha:X\to TF_\tau X$) and $\beta:Y\multimap \overline{F}^{*}Y$ (i.e. $\beta:Y\to TF^{*}Y$). We have:
\begin{align*}
\alpha^\bigstar = (m_X'\cdot \overline{F_\tau}\alpha)^{*}\cdot e_X' \qquad \text{ and }\qquad \beta^\divideontimes = (m_Y\cdot \overline{F}^{*}\beta)^{*}\cdot e_Y.
\end{align*} 
\end{itemize}  
\begin{example}
Let $T=\mathcal{P}$ and $F_\tau = \Sigma_\tau\times \mathcal{I}d$, $F^{*} =\Sigma^{*}\times \mathcal{I}d$. The morphism $h_X:\overline{\Sigma^{*}}X\multimap \overline{\Sigma_\tau}X$ is given by:
$$
h_X:\Sigma^{*}\times X \to \mathcal{P}(\Sigma_\tau \times X); (s,x)\mapsto \left \{\begin{array}{cc}\{(\tau,x)\} & \text{if } |s|=0, \\ \{(s,x)\} & \text{if }|s|=1,\\ \varnothing & \text{otherwise.} \end{array}\right.
$$
\end{example}

Consider any coalgebra $\alpha:X\multimap \overline{F_\tau}X$ and let $\underline{\alpha}:X\multimap \overline{F}^{*}X$ be given by $\underline{\alpha} = [\nu_X,e_X] \cdot \alpha$. Note that this is the same coalgebra as in the paragraph on monadic structure on $TF^{*}$ in Section \ref{subsection:monads}. Here, however, it is defined in terms of the composition in $\mathcal{K}l(T)$ and not $\mathsf{C}$, and all superscripts $^\sharp$ are dropped to simplify the notation. By commutativity of the last diagram above we have:
$$h_X\cdot \underline{\alpha} = h_X\cdot [\nu_X,e_X] \cdot \alpha = \alpha.$$ 
We will now try to compare bisimulations for $\alpha^{\bigstar}$ and $\underline{\alpha}^\divideontimes$. In case of labelled transition systems a relation is a bisimulation on $\alpha^{\bigstar}$ if and only if it is a bisimulation on $\underline{\alpha}^\divideontimes$. Below we verify how general is this statement and what conditions are required to be satisfied for it to remain true.

\begin{lemma}\label{lemma:stars}
Assume that for any $\phi:\overline{F}^{*}X\multimap \overline{F}^{*}X$ and $\psi:\overline{F_\tau}X\multimap \overline{F_\tau}X$ if  $\psi \cdot h_{X} = h_{X}\cdot \phi$ then  $\psi^{*} \cdot h_{X} = h_{X}\cdot \phi^{*}$. In this case we have $h_X\cdot \underline{\alpha}^\divideontimes = \alpha^{\bigstar}$.
\end{lemma}

\begin{remark}
Note that the assumption in Lemma \ref{lemma:stars} about the natural transformation $h$ is crucial even though $T$ is assumed to be an order saturation monad. Assumption (\ref{axioms:4}) in the definition of order saturatiom monad does not guarantee that $h$ satisfies the desired property since it is not in general of the form $h'^{\sharp}$ for some $h':F^{*}X\to F_\tau X$ in $\mathsf{C}$. However, if $T$ is a Kleene monad \cite{Brengos13,Gov} then this assumption is always satisfied. The powerset monad $\mathcal{P}$ is an example of a Kleene monad.
\end{remark}
The following theorem follows directly from the above lemma.
\begin{theorem}\label{theorem:bisimulation_sat_1}
Assume that for any $\phi:\overline{F}^{*}X\multimap \overline{F}^{*}X$ and $\psi:\overline{F_\tau}X\multimap \overline{F_\tau}X$ if  $\psi \cdot h_{X} = h_{X}\cdot \phi$ then  $\psi^{*} \cdot h_{X} = h_{X}\cdot \phi^{*}$. 
Any bisimulation on $\underline{\alpha}^\divideontimes$ is a bisimulation on $\alpha^{\bigstar}$.
\end{theorem}

Our aim now will be to prove the converse.

\begin{lemma}\label{lemma:two_stars_lemma}
We have
$
\underline{\alpha}^\divideontimes \leqslant (\underline{\alpha^\bigstar})^{\divideontimes}$.
\end{lemma}
\begin{remark}
Before we state the next result we have to make one essential remark. Note that the technical condition concerning the transformation $[\nu,e]$ in the lemma below would follow from $[\nu,e]$ being a monad morphism.  However, $[\nu,e]:\overline{F_\tau}\implies \overline{F}^{*}$ is \emph{not} a monad morphism. It does not satisfy the 2nd axiom of a monad morphism. To see this consider $T=\mathcal{P}$, $(\overline{\Sigma_\tau},m',e')$, $(\overline{\Sigma}^{*},m,e)$ as in Examples  \ref{example:LTS_monad} and \ref{example:LTS_monad_extended} and a visible label $a\in \Sigma$. We have
\begin{align*}
&[\nu_X,e_X]\cdot m'_X(a,a,x) = \varnothing \text{ and }  \\
&m_X\cdot [\nu_{\overline{\Sigma^{*}}X},e_{\overline{\Sigma^{*}}X}]\cdot \overline{\Sigma_\tau} [\nu_X,e_X] (a,a,x) =\{(aa,x)\}.
\end{align*}
\end{remark}
\begin{lemma}\label{lemma:sub_commute}
Assume $[\nu_X,e_X]\cdot m_X'\cdot \overline{F_\tau} \alpha \leqslant  m_X\cdot \overline{F}^{*} \underline{\alpha} \cdot [\nu_X,e_X]$. Then $\underline{\alpha^\bigstar} \leqslant \underline{\alpha}^\divideontimes$.
\end{lemma}

\begin{theorem}\label{theorem:star_star_bisimulation}
Let $\alpha$ satisfy the inequality from the assumptions of the previous statement.
Any bisimulation on ${\alpha^\bigstar}$ is a bisimulation on $\underline{\alpha}^\divideontimes$.
\end{theorem}
\begin{proof}
We have $\alpha \leqslant \alpha^\bigstar$ and hence $\underline{\alpha}\leqslant \underline{\alpha^\bigstar}$. This, together with Lemma \ref{lemma:sub_commute}, implies that  $\underline{\alpha}^{\divideontimes} \leqslant (\underline{\alpha^\bigstar})^\divideontimes\leqslant (\underline{\alpha}^\divideontimes)^\divideontimes=\underline{\alpha}^\divideontimes$.
Assume $X\stackrel{\pi_1}{\leftarrow} R\stackrel{\pi_2}{\rightarrow} X$ is a bisimulation on $\alpha^\bigstar$. It is also a bisimulation on $\underline{\alpha^\bigstar}$. Finally, since $\underline{\alpha}^\divideontimes = (\underline{\alpha^\bigstar})^\divideontimes$ the relation $R$ is a bisimulation on $\underline{\alpha}^\divideontimes$.
\end{proof}


\begin{theorem}\label{cor:weak_bisims}
Assume that cotupling $[-,-]$ in $\mathcal{K}l(T)$ is monotonic w.r.t. both arguments and the zero morphisms $0_{X,Y}:X\multimap Y$ are the least elements of the posets $Hom_{\mathcal{K}l(T)}(X,Y)$. Then any bisimulation on ${\alpha}^\bigstar$ is a bisimulation on $\underline{\alpha}^\divideontimes$.
\end{theorem}
\begin{remark}
The powerset monad $\mathcal{P}$ satisfies assumptions of the above theorem. It is worth mentioning that the $\mathcal{CM}$ monad used to model Segala systems does not satisfy them as the zero morphisms in $\mathcal{K}l(\mathcal{CM})$ are not least elements of the partially ordered hom-sets \cite{Jacobs08}. The monad $\mathcal{CM}$ deserves a separate treatment and we leave this for future research. 
\end{remark}

\section{Trace semantics for coalgebras with internal moves}
\label{section:trace_semantics}
The aim of this section is to present some ideas on how to approach the notion of trace semantics for structures with invisible moves. As mentioned before in order to distinguish the trace semantics for coalgebras with and without silent steps we will often use the term \emph{weak trace semantics} or \emph{trace semantics for structures with internal moves} to refer to the former.  

Before we go into details we start this section by recalling a basic example of  trace semantics for $\varepsilon$-NA's \cite{HopUll}. 
 
 \begin{definition}\label{definition:e-na_trace}
Given a non-deterministic automaton with $\varepsilon$-transitions $\alpha:X\to \mathcal{P}(\Sigma_\varepsilon\times X +  1)$ its \emph{trace semantics} is a morphism $\text{tr}_\alpha :X\to \mathcal{P}(\Sigma^{*})$ which maps any state $x\in X$ to the set of words over $\Sigma$ it accepts. To be more precise, for a word $w\in \Sigma^{*}$ we have $w\in \text{tr}_\alpha(x)$ provided that either $w=\varepsilon$ and $\checked \in \alpha(x)$ or $w=a_1\ldots a_n$ for $a_i\in \Sigma$ and there is $x'\in X$ such that 
$$x(\stackrel{\varepsilon}{\to})^{*}\circ\stackrel{a_1}{\to}\circ  (\stackrel{\varepsilon}{\to})^{*}\ldots (\stackrel{\varepsilon}{\to})^{*}\circ\stackrel{a_n}{\to}\circ  (\stackrel{\varepsilon}{\to})^{*}x'$$
with $\checked \in \alpha(x')$. 
\end{definition}
The above definition is an instance of what we call a ``bottom-up" approach towards trace semantics for non-deterministic automata with internal moves. This approach 
considers $\varepsilon$ steps as invisible steps that can wander around a structure freely. In other words, from our perspective $\varepsilon$-steps that are used in this definition are what they should be, i.e. are part of the unit of the $\varepsilon$-NA monad.
There is a second obvious approach towards defining trace semantics for $\varepsilon$-NA's. We call this approach ``top-down", since at first we treat $\varepsilon$ steps artificially as if they were standard visible steps. Given an $\varepsilon$-NA $\alpha:X\to \mathcal{P}(\Sigma_\varepsilon\times X +  1)$ we find its trace $\text{tr}'_\alpha:X\to \mathcal{P}((\Sigma\cup \{\varepsilon\})^{*})$ and then map all words from $(\Sigma\cup\{\varepsilon\})^{*}$ to words in $\Sigma^{*}$ by removing all occurrences of the $\varepsilon$ label. As a result we obtain the same trace as in Def. \ref{definition:e-na_trace}. Since in many cases we know how to find finite trace semantics for coalgebras with only visible steps \cite{HasJacSok} it is easy to generalize the ``top-down" approach to coalgebras with internal activities. This is exactly how authors of \cite{HasJacSok_jap,SilWester} do it in their papers. We, however, will present a bottom-up approach towards weak trace semantics that works for a large family of coalgebras whose type is a monad.   

\subsubsection{Coalgebraic view on weak trace semantics for $\varepsilon$-NA}
In this subsection we focus on coalgebras for the monad $\mathcal{P}(\Sigma^{*}\times X +  \Sigma^{*})$. Recall that by Example \ref{example_NA_translate} any $\varepsilon$-NA coalgebra $\alpha:X\to \mathcal{P}(\Sigma_\varepsilon \times X +  1)$ can be considered a $\mathcal{P}(\Sigma^{*}\times X +  \Sigma^{*})$-coalgebra. For simplicity and clarity of notation put $F=\Sigma\times \mathcal{I}d +  1$ and $F^{*} = \Sigma^{*}\times \mathcal{I}d +  \Sigma^{*}$. 
Let us list two basic facts concerning $\varepsilon$-NA monad:
\begin{itemize}
\item The lifting $\overline{F}^{*}:\mathcal{K}l(\mathcal{P})\to\mathcal{K}l(\mathcal{P})$ is locally continuous \cite{HasJacSok}.
\item The $\varepsilon$-NA monad  $\mathcal{P}F^{*}$ is $\textbf{Cppo}$-enriched. This follows by Theorem~\ref{theorem_Cppo_continuous}.
\end{itemize}

 For any $\alpha:X\multimap X$ in $\mathcal{K}l(\mathcal{P}F^{*})$ (i.e. $\alpha:X\to \mathcal{P}F^{*}X$) define the following mapping $\text{tr}_\alpha:X\multimap \varnothing$ (i.e. $\text{tr}_\alpha:X\to \mathcal{P}(\Sigma^{*})$):
$$\text{tr}_\alpha = \bigvee_{n\in \mathbb{N}} \perp \cdot \alpha^n,$$
where $\perp:X\multimap \varnothing$ is given by  $\perp:X\to \mathcal{P}(\Sigma^{*});x\mapsto \varnothing$ and $\cdot$ denotes the composition in $\mathcal{K}l(\mathcal{P}F^{*})$ as in Example \ref{example:na_composition}.
It is simple to see that $\text{tr}_\alpha$ is the least morphism in $Hom_{\mathcal{K}l(\mathcal{P}F^{*})}(X,\varnothing)=Hom_{\mathsf{Set}}(X,\mathcal{P}(\Sigma^{*}))$ satisfying $\text{tr}_\alpha = \text{tr}_\alpha \cdot \alpha$. In other words, $$\text{tr}_\alpha = \mu x. x\cdot \alpha.$$
Recursively, if we put $\text{tr}_0 = \perp$ and $\text{tr}_n = \text{tr}_{n-1} \cdot \alpha$ then $\text{tr}_\alpha = \bigvee_n \text{tr}_n$. 
\begin{example}
Let $\Sigma = \{a,b\}$ and let $\alpha:X\to \mathcal{P}(\Sigma_\varepsilon\times X +  1)$ be given by the following diagram ($\varepsilon$-labels are omitted).  We have  $\text{tr}_0:X\to \mathcal{P}(\Sigma^{*}),x\mapsto \varnothing$ and
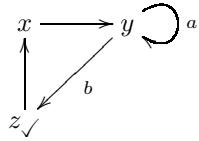
\begin{multicols}{2}
$
\xymatrix{
x \ar[r] & y\ar@(ur,dr)^a \ar[dl]^b \\
z_{\checked} \ar[u] &
}
$ 
\columnbreak
\begin{align*}
\text{tr}_1: & x\mapsto \varnothing, y\mapsto\varnothing, z\mapsto \{\varepsilon\},\\
\text{tr}_2: & x\mapsto \varnothing, y\mapsto\{b\}, z\mapsto \{\varepsilon\},\\
\text{tr}_3: & x\mapsto \{b\}, y\mapsto\{ab,b\}, z\mapsto \{\varepsilon\},\\
\text{tr}_4: & x\mapsto \{ab,b\}, y\mapsto\{aab,ab,b\}, \\ & z\mapsto \{b,\varepsilon\}
\end{align*}
\end{multicols}

\end{example}

The following result can be shown by straightforward verification.
\begin{theorem}\label{theorem:trace_semantics_least}
For any $\varepsilon$-NA coalgebra considered as $\mathcal{P}(\Sigma^{*}\times \mathcal{I}d +  \Sigma^{*})$-coalgebra the trace semantics morphism from Def. \ref{definition:e-na_trace} and the morphism $\text{tr}_\alpha$ above coincide.
\end{theorem}

\subsubsection{Weak coalgebraic trace semantics via fixed point operator} 

We see that for $\varepsilon$-NA's their weak trace semantics is obtained as the least fixed point of the assignment $x\mapsto x\cdot \alpha$ in $\mathcal{K}l(\mathcal{P}(\Sigma^{*}\times \mathcal{I}d+\Sigma^{*}))$. Interestingly, such a fixed point is not unique. 
\begin{example}
Let $\Sigma =\{a\}, X=\{x\}$ and let $\varepsilon$-NA $\alpha:X\to \mathcal{P}(\Sigma_\varepsilon\times X+1)$ be defined by the following diagram: $\xymatrix{x\ar@(ur,dr)^\varepsilon}$.  It is easy to check that the morphism $g:X\to \mathcal{P}(\Sigma^{*});x\mapsto \{a\}$ satisfies $g = g\cdot \alpha$ and it is not the least fixed point since the least fixed point is given by $\text{tr}_\alpha(x) = \varnothing$.
\end{example}

Here we generalize the ideas presented in the previous subsection to $T$-coalgebras. It should be noted at the very beginning that this section should serve as merely a starting point for future research.

Let us first focus on a known approach for defining trace semantics via coinduction in Kleisli category \cite{HasJacSok} and translating these results to our setting. In \cite{HasJacSok} the authors present trace semantics definition via coinduction for $TF$-coalgebras, where $T$ is a monad and $F$ satisfies some reasonable assumptions. In our setting however, we do not consider a special functor $F$ or in other words $F=\mathcal{I}d$ and our coalgebras are $T$-coalgebras.  Consider the category $\mathcal{K}l(T)_{\mathcal{I}d}$ of $\mathcal{I}d$-coalgebras in $\mathcal{K}l(T)$.
Note that any $T$-coalgebra $\alpha:X\to TX$ is $\alpha:X\multimap X$ and is a member of $\mathcal{K}l(T)_{\mathcal{I}d}$. Based on the approach from \cite{HasJacSok} trace semantics of $\alpha$ should be obtained via coinduction in  $\mathcal{K}l(T)_{\overline{F}}$. In our setting however, i.e. when $F=\mathcal{I}d$, the category $\mathcal{K}l(T)_{\mathcal{I}d}$ rarely admits the terminal object. For instance if we consider our $\varepsilon$-NA monad $\mathcal{P}(\Sigma^{*}\times \mathcal{I}d  +  \Sigma^{*})$, the category of $\mathcal{I}d$-coalgebras  $\mathcal{K}l(\mathcal{P}(\Sigma^{*}\times \mathcal{I}d +  \Sigma^{*}))_{\mathcal{I}d}$ has no terminal object. However, it still makes sense to talk about trace for coalgebras for the monad $\mathcal{P}(\Sigma^{*}\times \mathcal{I}d  +  \Sigma^{*})$. We did it via the least fixed point of the assignment $x\mapsto x\cdot \alpha$. In the general case we do it via uniform fixed point operator \cite{SimPlot}. 

Assume that $\mathsf{C}$ is a category with the initial object $0$ (this object is also initial in $\mathcal{K}l(T)$). A \emph{fixed point operator} $\text{f}$ on $\mathcal{K}l(T)$ is a family of morphisms:
$$
\text{f}:Hom_{\mathcal{K}l(T)}(X,X)\to Hom_{\mathcal{K}l(T)}(X,0)
$$
satisfying $\text{f}(\alpha) \cdot \alpha = \text{f}(\alpha)$ for any $\alpha:X\multimap X$. 
A fixed point operator $\text{f}$ on $\mathcal{K}l(T)$ is \emph{uniform} w.r.t. $(-)^\sharp:\mathsf{C}\to \mathcal{K}l(T)$ \cite{SimPlot} if 
$$h^\sharp \cdot \alpha = \beta \cdot h^\sharp \implies \text{f}(\beta) \cdot h^\sharp = \text{f}(\alpha)$$ 
for any $\alpha:X\multimap X$, $\beta:Y\multimap Y$ in $\mathcal{K}l(T)$ and $h:X\to Y$ in $\mathsf{C}$.
Coalgebraically speaking, the premise of the above implication says that the morphism $h$ is a homomorphism between coalgebras $\alpha:X\to TX$ and $\beta:Y\to TY$ in $\mathsf{C}_T$.  We call a uniform fixed point operator on $\mathcal{K}l(T)$  a \emph{coalgebraic trace operator} and we denote it by $\text{tr}_{(-)}$. 

\begin{theorem}\label{theorem:cppo_coalgebraic_trace}
Assume that $\mathcal{K}l(T)$ is a $\textbf{Cppo}$-enriched category and assume that for any $f:X\to Y$ in $\mathsf{C}$ we have $\perp\cdot f^\sharp = \perp$. For $\alpha:X\multimap X$ define $\text{tr}_\alpha:X\multimap 0$ by $\text{tr}_\alpha= \mu x.(x\cdot \alpha) = \bigvee_{n\in \mathbb{N}}\perp\cdot \alpha^n$.
Then $\text{tr}_{(-)}$ is a coalgebraic trace operator on $\mathcal{K}l(T)$.
\end{theorem}

It may not be instantly clear for the reader why we choose uniformity as a property of a coalgebraic trace operator. Uniformity is a powerful notion which, in some forms, determines the least fixed point to be the unique uniform fixed point operator \cite{SimPlot}.  
For the $\varepsilon$-NA monad $\mathcal{P}(\Sigma^{*}\times \mathcal{I}d+\Sigma^{*})$ the least fixed point operator is a uniform fixed point operator w.r.t. $$^\sharp:\mathsf{Set}\to \mathcal{K}l(\mathcal{P}(\Sigma^{*}\times \mathcal{I}d+\Sigma^{*})).$$ However, as we will see further on (Theorem \ref{theorem:traced_fixed_point} and Example \ref{example:NA_trace}), it is uniform also with respect to a richer category than $\mathsf{Set}$, namely, it is uniform w.r.t.:
$$^\sharp:\mathcal{K}l(\mathcal{P}(\Sigma^{*}\times \mathcal{I}d))\to \mathcal{K}l(\mathcal{M}_1)\cong\mathcal{K}l(\mathcal{P}(\Sigma^{*}\times \mathcal{I}d+\Sigma^{*})).$$

Uniqueness of a uniform fixed point operator on $\mathcal{K}l(T)$ can be imposed by inital algebra = final coalgebra coincidence in the base category $\mathsf{C}$ \cite{SimPlot}. This coincidence is the core of generic coalgebraic trace semantics theory \cite{HasJacSok}. This is why we believe that the uniform fixed point operators can and will serve as an extension of the generic coalgebraic trace semantics to weak trace semantics.

 We end this section with a result that links weak trace semantics for $\varepsilon$-NA's to uniform traced monoidal categories in the sense of Joyal et al. \cite{JSV}. However, instead of a uniform categorical trace operator on a monoidal category with binary coproducts and initial object we will equivalently work with a uniform Conway operator \cite{Hasegawa,JSV}. The following theorem (modulo the uniformity) can be found in \cite{BentonHyland}.
\begin{theorem}\label{theorem:traced_fixed_point}
Assume $\mathsf{C}$ is equipped with a uniform Conway operator $$(-)_{X,A}^{\dagger}:Hom(X,X +  A)\to Hom(X,A).$$ Let $A$ be an object in $\mathsf{C}$ and $\mathcal{M}_A = \mathcal{I}d +  A$ the exception monad on $\mathsf{C}$. Then the operator $\text{tr}_{(-)}:Hom_{\mathcal{K}l(\mathcal{M}_A)}(X,X)\to Hom_{\mathcal{K}l(\mathcal{M}_A)}(X,0)$ defined by $\text{tr}_\alpha = \alpha^\dagger$ for $\alpha:X\to X +  A$ in $\mathsf{C}$ (or equivalently $\alpha:X\multimap X$ in $\mathcal{K}l(\mathcal{M}_A)$) is a coalgebraic trace operator on the category $\mathcal{K}l(\mathcal{M}_A)$ which is uniform w.r.t. $^\sharp:\mathsf{C}\to \mathcal{K}l(\mathcal{M}_A)$.
\end{theorem}

\begin{example}\label{example:NA_trace}
The $\varepsilon$-NA's and their trace semantics fits into the above setting since the $\varepsilon$-NA monad satisfies:
$$\mathcal{P}(\Sigma^{*}\times \mathcal{I}d  +  \Sigma^{*}) \cong \mathcal{P}(\Sigma^{*}\times (\mathcal{I}d +  1)).$$ Hence,  if we put $T=\mathcal{P}(\Sigma^{*}\times \mathcal{I}d)$ to be the free LTS monad then the $\varepsilon$-NA monad is given by $T(\mathcal{I}d +  1)=T\mathcal{M}_1$. Since the free LTS monad $\mathcal{P}(\Sigma^{*}\times \mathcal{I}d) \cong \mathcal{P}(\Sigma^{*})^{\mathcal{I}d}$ is an example of a quantale monad \cite{Jacobs10} on $\mathsf{Set}$ its Kleisli category $\mathcal{K}l( \mathcal{P}(\Sigma^{*}\times \mathcal{I}d) )$ with binary coproducts and initial object is equipped with a uniform Conway operator (or equivalently a uniform categorical trace operator) \cite{Hasegawa,Jacobs10}. Therefore, if we put $\mathsf{C}=\mathcal{K}l(\mathcal{P}(\Sigma^{*}\times \mathcal{I}d))$ then the Kleisli category for the exception monad $\mathcal{M}_1 = \mathcal{I}d  +  1$ defined on $\mathsf{C}$ is isomorphic to the Kleisli category for $\varepsilon$-NA monad, i.e. $\mathcal{K}l(\mathcal{M}_1) \cong \mathcal{K}l(\mathcal{P}(\Sigma^{*}\times \mathcal{I}d +  \Sigma^{*}))$. The analysis of the Conway operator for the Kleisli category for the monad $\mathcal{P}(\Sigma^{*}\times \mathcal{I}d)$ \cite{Jacobs10} leads to a conclusion that $\text{tr}_\alpha$ obtained for $\varepsilon$-NA's via Theorem \ref{theorem:traced_fixed_point} is exactly the least fixed point operator we introduced in the previous subsection.
\end{example}

To conclude, when allowing invisible steps into our setting, i.e. considering coalgebras over monadic types, weak trace semantics becomes a categorical fixed point operator. Moreover, as the above example states, there is a strong connection between coalgebraic trace operator for $\varepsilon$-NA coalgebras and traced monoidal categories. Although traced categories have been studied from coalgebraic perspective in \cite{Jacobs10} they were considered a special instance of the generic coalgebraic trace theory. With Example \ref{example:NA_trace} at hand we believe that it should be the other way around in many cases, i.e. coalgebraic trace semantics for coalgebras with internal moves is a direct consequence of the fact that certain Kleisli categories are traced monoidal categories.

\section{Weak bisimulation and weak trace semantics}\label{section:weak_and_trace}

 We have shown that two behavioural relations, namely, weak bisimulation and weak trace equivalence can be defined using fixed points of certain maps. In case of trace equivalence this map is given by $x\mapsto x\cdot \alpha$, in case of weak bisimulation it is $x\mapsto 1\vee x\cdot \alpha$. We see that both equivalences should be considered individually, as they require different assumptions. Yet, in a restrictive enough setting we should be able to compare these notions at once. Indeed, in the setting of monads whose Kleisli category has hom-sets being complete join semilattices and whose composition preserves all non-empty joins, it is possible for us to talk about three behavioural equivalences at once, namely, weak trace semantics, weak bisimilarity and bisimilarity. In this case we can prove the following. 

\begin{theorem}\label{theorem:weak_trace_strong}
Let $T$ be a monad as above and let $\perp = \perp \cdot f^\sharp \text{ for any } f:X\to Y$ in $\mathsf{C}$. A strong bisimulation on $\alpha:X\to TX$ is also a weak bisimulation on $\alpha$. Moreover, if we define the trace map to be $\text{tr}_\alpha = \mu x. x\cdot\alpha$ then $\text{tr}_\alpha = \text{tr}_{\alpha^{*}}$. In other words, weak bisimilarity implies weak trace equivalence.
\end{theorem}
 
\section{Summary and future work}\label{section:summary}

This paper shows that coalgebras with internal moves can be understood as coalgebras over a type which is a monad. We believe that such a treatment makes formulation of many different properties and behavioural equivalences simpler. It is natural to suspect that many other types of different behavioural equivalences can be translated into the coalgebraic setting this way. One of these is dynamic bisimulation \cite{MonSass} which should be obtained as a strong bisimulation on $\mu x. (\alpha  \vee  x\cdot \alpha)$ (i.e. a transitive closure of $\alpha$). We believe that this paper may serve as a starting point for a larger project to translate some of the equivalences from van Glabbeek's spectrum of different equivalences for state-based systems with silent labels \cite{Glabbeek,Sangiorgi11} into the setting of coalgebras with internal activities.

Finally, as mentioned in Section \ref{section:trace_semantics} we should aim at extending the coalgebraic trace semantics theory for systems without internal transitions  \cite{HasJacSok} to systems with silent moves. Uniform fixed point operator could serve as such an extension. Moreover, we should build a more traced monoidal category oriented theory of coalgebraic traces and refer it to known results for generic coalgebraic trace. 

\paragraph{Acknowledgements} I would like to thank Alexandra Silva for inspiring me with the literature on categorical fixed points. I am also very grateful to anonymous referees for various comments and remarks that hopefully made this work more interesting and easier to follow.

\newpage

\section*{Appendix}

\subsubsection{Section \ref{section:LTS_coalgebraically}}

\begin{proof}[Lemma \ref{lemma:stars}] Consider $\alpha:X\multimap \overline{F_\tau} X$. Then $\underline{\alpha}:X\multimap \overline{F}^{*}X$. Moreover, the following diagram commutes in $\mathcal{K}l(T)$:
$$
\xymatrix@-1pc{
\overline{F}^{*}X \ar@{-o}[d]_{h_X}\ar@{-o}[r]^{\overline{F}^{*}\underline{\alpha}} & \overline{F}^{*} \overline{F}^{*} X\ar@{-o}[d]^{h_{\overline{F}^{*}X}} \ar@{-o}[r]^{m_X} & \overline{F}^{*} X \ar@{-o}[dd]^{h_X}\\
\overline{F_\tau}X\ar@{-o}[r]^{\overline{F}_\tau \underline{\alpha}} \ar@{-o}[dr]_{\overline{F_\tau}\alpha}& \overline{F_\tau}\overline{F}^{*}X \ar@{-o}[d]^{\overline{F_\tau}h_{X}} \\
& \overline{F_\tau}\overline{F_\tau}X\ar@{-o}[r]_{m'_X} & \overline{F_\tau}X 
}
$$
By our assumption it follows that $h_X\cdot (m_X\cdot \overline{F}^{*}\underline{\alpha})^{*} = (m'_X\cdot \overline{F_\tau}\underline{\alpha})^{*}\cdot h_X$. Hence,
\begin{align*}
&h_X\cdot  \underline{\alpha}^\divideontimes =  h_X\cdot (m_X\cdot \overline{F}^{*}\underline{\alpha})^{*}\cdot e_X = \\ &(m'_X\cdot \overline{F_\tau}{\alpha})^{*}\cdot h_X\cdot e'_X = (m'_X\cdot \overline{F_\tau}{\alpha})^{*}\cdot e'_X = \alpha^{\bigstar}.
\end{align*}
\end{proof}

\begin{proof}[Theorem \ref{theorem:bisimulation_sat_1}]
Let $X\stackrel{\pi_1}{\leftarrow} R\stackrel{\pi_2}{\to}X$ be a bisimulation on $\underline{\alpha}^\divideontimes$. This means there is $\gamma:Z\to TF^{*}Z$ and two homomorphisms $f$ and $g$ in $\mathsf{C}_{TF^{*}}$ between $\underline{\alpha}^\divideontimes$ and $\gamma$ such that $R$ with $\pi_1$ and $\pi_2$ is a pullback of a suitable diagram. Since $TF^{*}$ is order saturation monad the morphisms $f$ and $g$ are also homomorphisms between $\underline{\alpha}^\divideontimes$ and $\gamma^\divideontimes$. Since the monad morphism $h$ is a natural transformation it follows that $f$ and $g$ are also homomorphisms between $\alpha^\bigstar = h_X\cdot \underline{\alpha}^\divideontimes$ and $h_Z\cdot \gamma^\divideontimes$ which completes the proof.
\end{proof}

\begin{proof}[Lemma \ref{lemma:two_stars_lemma}]
This follows directly from the fact that $TF_\tau$ and $TF^{*}$ are order saturation monads with saturation operators given by $(-)^\bigstar$ and $(-)^\divideontimes$ respectively. Indeed, this implies that $\alpha \leqslant \alpha^\bigstar$. Hence, $\underline{\alpha}\leqslant \underline{\alpha^\bigstar}$ and $\underline{\alpha}^\divideontimes \leqslant \underline{\alpha^\bigstar}^{\divideontimes}$.
\end{proof}

\begin{proof}[Lemma \ref{lemma:sub_commute}]
Consider $\alpha:X\multimap \overline{F_\tau}X$ and assume that $$[\nu_X,e_X]\cdot m'_X\cdot \overline{F_\tau} \alpha \leqslant  m_X\cdot \overline{F}^{*} \underline{\alpha} \cdot [\nu_X,e_X].$$ 
Since $[\nu_X,e_X] = [\nu'_X,\varepsilon_X]^\sharp$, where $\varepsilon:\mathcal{I}d\implies F^{*}$ is the unit of the free monad $F^{*}$ in $\mathsf{C}$ and $\nu'_X:F\implies F^{*}$ arises by freeness of $F^{*}$ in $\mathsf{C}$, by property (\ref{axioms:4}) of the definition of ordered saturation monad we infer that $[\nu_X,e_X]\cdot (m'_X\cdot \overline{F_\tau} \alpha)^{*} \leqslant  (m_X\cdot \overline{F}^{*} \underline{\alpha})^{*} \cdot [\nu_X,e_X]$. Hence, 
\begin{align*}
&\underline{\alpha^\bigstar}=[\nu_X,e_X]\cdot (m_X'\cdot \overline{F_\tau} \alpha)^{*} \cdot e'_X \leqslant  (m_X\cdot \overline{F}^{*} \underline{\alpha})^{*} \cdot [\nu_X,e_X]\cdot e'_X  =\\
& (m_X\cdot \overline{F}^{*} \underline{\alpha})^{*} \cdot e_X = \underline{\alpha}^\divideontimes. 
\end{align*}
\end{proof}

\begin{proof}[Theorem \ref{cor:weak_bisims}]
Our aim is to prove the inequality from the previous theorem. By our assumptions we have:
$$
\xymatrix{
\overline{F}X \ar@{-o}[r]^{0} \ar@{-o}[d]_{\nu_{X} } & X\ar@{-o}[dl]^{e_X}_{\geq} \\
\overline{F}^{*}X & 
}
 \xymatrix{
\overline{F}X +  X \ar@{-o}[r]^{[0,id]} \ar@{-o}[d]_{[\nu_{X},e_X] } & X\ar@{-o}[dl]^{e_X}_{\geq} \\
\overline{F}^{*}X & 
}
 \xymatrix{
\overline{F}\overline{F_\tau}X \ar@{-o}[r]^{\overline{F}[0,id]} \ar@{-o}[d]_{\overline{F}[\nu_{X},e_X] } & \overline{F}X\ar@{-o}[dl]^{\overline{F}e_X}_{\geq} \\
\overline{FF^{*}}X & 
}
$$
Hence,
$$
\xymatrix{
\overline{F}\overline{F_\tau}X \ar@{-o}[r]^{\overline{F}[0,id]} \ar@{-o}[d]_{\overline{F}[\nu_X,e_X]} 
& \overline{F} X\ar@{-o}[dl]^{\overline{F}e_X}_\geq \ar@{-o}[d]^{\nu_{X}}\\
\overline{FF^{*}}X \ar@{-o}[d]_{\nu_{\overline{F}^{*}X}}& \overline{F}^{*}X\ar@{-o}[dl]^{\overline{F}^{*}e_X} \\
\overline{F^{*} F^{*}}X \ar@{-o}[r]_{m_X} & \overline{F}^{*}X\ar@{=}[u]
}
\text{ and }
\xymatrix{
\overline{F_\tau}X \ar@{=}[r] \ar@{-o}[d]_{[\nu_X,e_X]} 
& \overline{F_\tau} X \ar@{-o}[dd]^{[\nu_{X},e_X]}\\
\overline{F}^{*}X \ar@{-o}[d]_{e_{\overline{F}^{*}X}}&  \\
\overline{F^{*} F^{*}}X \ar@{-o}[r]_{m_X} & \overline{F}^{*}X\ar@{=}[lu]
}
$$
Since cotupling preserves order the rightmost rectangle in the following diagram op-lax commutes:
$$
\xymatrix{
&\overline{F_\tau}\overline{F_\tau}X \ar@{-o}[r]^{[\overline{F}[0,id], id]}_{m'_X} \ar@{}[ddr]|\geq \ar@{-o}[d]|{\overline{F_\tau}[\nu_{X},e_X]} & \overline{F}_\tau X \ar@{-o}[dd]^{[\nu_{X},e_{X}]}\\
\overline{F_\tau}X\ar@{-o}[d]_{[\nu_{X},e_{X}]}\ar@{-o}[r]|{\overline{F_\tau}\underline{\alpha}}\ar@{-o}[ur]^{\overline{F_\tau}\alpha} &\overline{F_\tau F^{*}}X \ar@{-o}[d]|{[\nu_{\overline{F}^{*}X},e_{\overline{F}^{*}X}]} \\
\overline{F}^{*}X\ar@{-o}[r]_{\overline{F}^{*}\underline{\alpha}} & \overline{F^{*} F^{*}}X \ar@{-o}[r]_{m_X} & \overline{F}^{*}X
}
$$
Hence, for any $\alpha:X\multimap \overline{F_\tau}X$ we have:
$$[\nu_X,e_X]\cdot m_X'\cdot \overline{F_\tau} \alpha \leqslant  m_X\cdot \overline{F}^{*} \underline{\alpha} \cdot [\nu_X,e_X].$$
This together with Theorem \ref{theorem:star_star_bisimulation} proves the assertion.
\end{proof}

\subsubsection{Section \ref{section:weak_and_trace}}

\begin{proof}[Theorem \ref{theorem:weak_trace_strong}]
We only need to show  the 2nd part of the statement as the first follows directly by the fact that $T$ is an order saturation monad.  The inequality $\text{tr}_\alpha\leqslant \text{tr}_\alpha^{*}$ holds by the fact that $\alpha \leqslant \alpha^{*}$ and that for any $\alpha\leqslant \beta$ we have $\text{tr}_\alpha\leqslant \text{tr}_\beta$. This follows directly by  $$\text{tr}_\alpha =\bigvee_n \perp \cdot \alpha^{n} \leqslant \bigvee_n \perp\cdot \beta^n = \text{tr}_\beta.$$ 
To see the inverse inequality is true note that 
$$\text{tr}_\alpha \cdot \alpha^{*} = (\bigvee_{n} \perp\cdot \alpha^n)\cdot \bigvee_{m=0,1,\ldots} \alpha^m = \bigvee_n \bigvee_{m=0,1,\ldots}  \perp \cdot \alpha^n\cdot \alpha^m= \text{tr}_\alpha .$$
Hence, $\text{tr}_\alpha$ is a fixed point of $x\mapsto x\cdot \alpha^{*}$. By the fact that $\text{tr}_{\alpha^{*}}$ is the least such element we get that $\text{tr}_{\alpha^{*}}\leqslant \text{tr}_\alpha$.
\end{proof}

\end{document}